\documentclass[a4paper, twoside, 11pt, reqno]{amsart}

\usepackage[T1]{fontenc}
\usepackage[english]{babel}
\usepackage[a4paper, hmargin=2.75cm, vmargin=3cm]{geometry}
\usepackage{lmodern, microtype, setspace, csquotes, enumitem, moreenum, comment}
\usepackage{amsmath, amsfonts, amssymb, amsthm, mathtools, mathrsfs, tikz}
\usepackage[foot]{amsaddr}
\usepackage{hyperref}
\usepackage{xurl}
\usepackage[capitalize]{cleveref}
\hypersetup{colorlinks=true, allcolors=blue!50!black, breaklinks=true}

\setdisplayskipstretch{2}

\setlist[1]{labelindent=\parindent, leftmargin=*, align=left, topsep=8pt, itemsep=4pt, parsep=0pt}
\SetEnumitemKey{env}{wide=\parindent, label=\normalfont\bfseries(\arabic*), ref=\normalfont(\arabic*), nosep}

\usetikzlibrary{arrows, nfold, decorations.pathreplacing}
\colorlet{shadecolor}{gray!20}

\renewcommand{\subset}{\subseteq}
\renewcommand{\supset}{\supseteq}
\renewcommand{\Gamma}{\varGamma}
\renewcommand{\Omega}{\varOmega}
\renewcommand{\Phi}{\varPhi}
\renewcommand{\Psi}{\varPsi}

\newcommand{\ce}{\coloneqq}
\newcommand{\ec}{\eqqcolon}
\newcommand{\wt}[1]{\widetilde{#1}}
\newcommand{\oln}[1]{\overline{#1}}
\newcommand{\N}{\mathbb{N}}

\newcommand{\R}{\mathbb{R}}
\newcommand{\C}{\mathbb{C}}
\newcommand{\SFR}{\mathsf{R}}
\newcommand{\SFS}{\mathsf{S}}
\newcommand{\SFT}{\mathsf{T}}

\newcommand{\MF}{\mathcal{F}}
\newcommand{\MH}{\mathcal{H}}
\newcommand{\MO}{\mathcal{O}}
\newcommand{\MSH}{\mathscr{H}}
\newcommand{\MSN}{\mathscr{N}}
\newcommand{\LO}{\mathscr{L}}
\newcommand{\BO}{\mathscr{B}}
\newcommand{\DM}{\mathscr{S}}
\newcommand{\Eig}{\mathrm{Eig}}
\newcommand{\ii}{\mathrm{i}}
\newcommand{\ee}{\mathrm{e}}
\newcommand{\1}{\mathop{}\!\mathbf{1}}
\newcommand{\id}{\mathop{}\!\mathrm{Id}}
\newcommand{\diff}{\mathop{}\!\mathrm{d}}
\newcommand{\vol}{\mathrm{vol}}
\newcommand{\area}{\mathrm{area}}
\newcommand{\dist}{\mathrm{dist}}
\newcommand{\diam}{\mathrm{diam}}
\newcommand{\dom}{\mathrm{dom}}
\newcommand{\ran}{\mathrm{ran}}
\newcommand{\EV}{\mathbb{E}}
\newcommand{\set}[2][]{#1\{{#2}#1\}}
\newcommand{\Set}[1]{\left\{{#1}\right\}}
\newcommand{\abs}[2][]{#1\vert{#2}#1\vert}
\newcommand{\norm}[2][]{#1\Vert{#2}#1\Vert}

\newcommand{\braket}[2][]{#1\langle{#2}#1\rangle}

\newcommand{\od}[3][]{\frac{\diff^{#1} #2}{\diff #3^{#1}}}
\newcommand{\sint}{\mathrm{int}}
\newcommand{\seff}{\mathrm{eff}}
\newcommand{\bdot}{\braket{\cdot, \cdot}}
\DeclareMathOperator{\tr}{tr}

\theoremstyle{definition}
\newtheorem{assumption}{Assumption}
\newtheorem{definition}{Definition}[section]
\newtheorem{remark}[definition]{Remark}
\newtheorem{remarks}[definition]{Remarks}
\newtheorem{example}[definition]{Example}
\newtheorem{examples}[definition]{Examples}

\theoremstyle{plain}
\newtheorem{theorem}[definition]{Theorem}
\newtheorem{proposition}[definition]{Proposition}
\newtheorem{lemma}[definition]{Lemma}
\newtheorem{corollary}[definition]{Corollary}
\newtheorem*{theorem*}{Theorem}

\AddToHook{env/remark/begin}{\crefalias{definition}{remark}}
\AddToHook{env/remarks/begin}{\crefalias{definition}{remarks}}
\AddToHook{env/example/begin}{\crefalias{definition}{example}}
\AddToHook{env/examples/begin}{\crefalias{definition}{examples}}
\AddToHook{env/theorem/begin}{\crefalias{definition}{theorem}}
\AddToHook{env/proposition/begin}{\crefalias{definition}{proposition}}
\AddToHook{env/lemma/begin}{\crefalias{definition}{lemma}}
\AddToHook{env/corollary/begin}{\crefalias{definition}{corollary}}

\crefname{assumption}{Assumption}{Assumptions}
\crefname{remarks}{Remark}{Remarks}
\crefname{examples}{Example}{Examples}
\crefname{appsec}{Appendix}{Appendices}

\begin{document}

\title[Open quantum systems with varying particle number]{Physico-mathematical model of open quantum systems with varying particle number}

\author[B. M. Reible]{Benedikt M. Reible}
\author[L. Delle Site]{Luigi Delle Site}
\address{Institute of Mathematics, Freie Universität Berlin, Arnimallee 6, 14195 Berlin, Germany}
\email{benedikt.reible@fu-berlin.de,luigi.dellesite@fu-berlin.de}

\begin{abstract}
    We derive the effective Hamiltonian $H - \mu N$ for open quantum systems with varying particle number from first principles. Our derivation relies on the general postulates of quantum statistical mechanics and on some specific assumptions regarding the size of the open system and the range of the interaction with the reservoir. Within this model, we also show that the effective Hamiltonian is unique up to a constant. In the course of our argument, we establish a rigorous version of the surface-to-volume ratio approximation, which is routinely used in an empirical form in statistical physics, and we show that the Hilbert space for systems with varying particle number must be isomorphic to Fock space. Together, these findings establish a solid physico-mathematical model for open quantum systems with varying particle number.
\end{abstract}

\maketitle

\section{Introduction}

\subsection{Prolegomenon}

Realistic quantum many-body systems are never truly isolated but rather exchange energy and possibly matter with their surroundings. Due to their practical relevance, for example, in the development of modern quantum technological devices, the theory of open quantum systems is a cornerstone of present-day theoretical and applied research \cite{AttalJoyePillet2006, Rotter2015, Vacchini2024}. While there are different elaborate mathematical frameworks in which to study these systems, e.g., the theory of operator algebras \cite{Derezinski2026} or Markovian quantum master equations \cite{Merkli2020}, numerous challenges regarding their theoretical description remain, especially for systems with varying particle number \cite{delRazo2025}. A deeper mathematical understanding of these systems is, however, not only desirable from a purely mathematical point of view, but it can also help, at first instance, to improve current state-of-the-art simulation techniques, which, in turn, can lead to new directions for the experimental realization of quantum systems for technology \cite{DelleSite2018}.

For an open quantum system with varying particle number, it is commonly accepted that the system's Hilbert space is given by Fock space and that its Hamiltonian $H$ should be extended by the term $- \mu N$, with $\mu$ the chemical potential and $N$ the particle number operator, resulting in the following \emph{effective Hamiltonian}:
\begin{equation}\label{eq:Heff}
    H_\seff = H - \mu N \ .
\end{equation}
In the physics literature, this operator is obtained as follows: one starts with the Hamiltonian of the open system, makes an \emph{ansatz} for the probability to find $n$ particles in the system in analogy to the canonical ensemble, and then obtains the grand canonical density
\begin{equation*}
    \rho_\mathrm{gc} \sim \ee^{- \beta (H - \mu N)}
\end{equation*}
via a Taylor expansion of that probability \cite[Sect. 2.4.1]{Nolting2018}, \cite[Sect. 2.7.2]{Schwabl2006}. Finally, one \emph{defines} the operator in the exponent of $\rho_\mathrm{gc}$ to be \emph{the} Hamiltonian of the open system with varying particle number \cite[p. 29]{LandauLifshitz9}, again by alluding to an analogy with the canonical density $\rho_\mathrm{c} \sim \ee^{- \beta H}$. The issue with this line of reasoning lies in the fact that the operator $H_\seff$ itself is, unlike the density $\rho_\mathrm{gc}$, not obtained from first principles but from the derived quantity $\rho_\mathrm{gc}$ by an empirical argument; see \cref{fig:illustrationProblem} for an illustration of this situation. In this article, we will give a mathematically solid derivation of the operator $H_\seff$ from first principles---starting from the total Hamiltonian of an open system coupled to a reservoir, introducing basic physical assumptions and approximations, and finally showing how the chemical potential $\mu$ naturally emerges---without making the detour through the density $\rho_\mathrm{gc}$. We will also argue that, within the physico-mathematical model staked out by the assumptions and approximations, $H_\seff$ is unique up to additive constants.

The use of the effective Hamiltonian \eqref{eq:Heff} in physics can be traced back at least to a conjecture by N.~N.~Bogoliubov in the field of condensed matter physics \cite{Bogoliubov1958}, where it has been used ever since; see, for example, Ref.~\cite[p. 314]{Zagrebnov2001} and references therein. Also in mathematical physics, it has become common practice to analyze functional-analytic properties of the operator $H_\seff$ and to build rigorous models with its help; see, e.g., Ref.~\cite[passim]{BratteliRobinson1981}. In neither of these fields, though, a derivation of $H_\seff$ from first principles is provided. Recently, however, a semi-empirical, first-principles derivation of \cref{eq:Heff} was proposed, obtaining the effective Hamiltonian for the open system by tracing out the degrees of freedom of the reservoir from the von Neumann equation of the combined system \cite{DelleSite2024jpa}. While this study put forward important ideas (on which we will, in part, build our analysis), the approach also relied on some empirical assertions and formal manipulations. Here we intend to eliminate these empirical inferences and make the derivation more rigorous. This will bridge the gap between studies concerned with the mathematical foundations of quantum mechanics, i.e., those dealing with general Hamiltonians and their properties, and studies focusing on specific mathematical problems for grand canonical systems, where the Hamiltonian $H_\seff$ is taken for granted.

\begin{figure}
    \centering
    \begin{tikzpicture}[line width=1pt, inner sep=6pt, text width=4cm, align=center]
        \node (FP) at (0,0) {\emph{First principles}};
        \node (SP) at (7,0) {\emph{Ensemble level}};

        \node (HS) at (0,-1.5) {Microscopic theory, \\[4pt] Hamiltonian $H$};
        \node (CE) at (7,-1.5) {Statistical theory, \\[4pt] $\rho_\mathrm{c} \sim \ee^{- \beta H}$};
        \node (GC) at (7,-5.4) {Extended theory, \\[4pt] $\rho_\mathrm{gc} \sim \ee^{- \beta (H - \mu N)}$};
        \node (HE) at (0,-5.4) {Effective Hamiltonian, \\[4pt] $H_\seff = H - \mu N$};
       
        \draw [->, very thick, dotted] (HS) -- (CE);
        \draw [->, very thick, dotted] (CE) -- (GC);
        \draw [->, very thick, dotted, color=red, transform canvas={yshift=0.25cm}] (GC) -- (HE);
        \draw [->, transform canvas={yshift=-0.25cm}] (HE) -- (GC);
        \draw [->] (HS) -- (HE);
        
        \draw [dashed] (3.5,-7) -- (3.5,0.5);
    \end{tikzpicture}
    \caption{Relation between first-principles and ensemble-level quantities. The dotted arrows represent the standard deductions in statistical physics, and the red inference is the problematic, empirical one since a first-principles quantity is obtained from an ensemble. The solid arrows depict the new derivations proposed in this text.}
    \label{fig:illustrationProblem}
\end{figure}

\subsection{Main results and outline of the paper}

To derive an expression like \eqref{eq:Heff}, we will have to analyze the structure of open quantum systems with varying particle number in a broader context and establish some preparatory propositions.

First, in \cref{sec:setup}, we introduce the general mathematical setup and basic physical assumptions, and we discuss some preliminary observations about computing expectation values of unbounded operators with respect to mixed states (density operators). In particular, we extend a well-known formula for the partial trace to unbounded operators and specific compatible density operators (\cref{lem:partialTrace}).

In \cref{sec:STVR}, we proceed by introducing specific assumptions on the size of the open system and the range of the system--reservoir interaction, and we prove a rigorous version of the so-called \emph{surface-to-volume ratio approximation} (\cref{pro:STVR}); this approximation is commonly used in statistical physics when deriving the canonical or grand canonical ensemble, and it consists of neglecting the system--reservoir interaction term $H_\sint$ in the Hamiltonian $H_\SFT = H_\SFS \otimes \id_\SFR + \id_\SFS \otimes H_\SFR + H_\sint$ of the combined system. Informally speaking, our result can be summarized as follows:
\begin{itemize}
    \item \itshape Under physically justified assumptions on the size of the open system and the range of the interaction, it holds that $\EV_\rho [H_\SFT] = \EV_\rho [H_\SFS \otimes \id_\SFR + \id_\SFS \otimes H_\SFR] + \MO(\delta^2)$ for suitable density operators $\rho$, where $\delta$ denotes the effective interaction distance.
\end{itemize}
The proof of this assertion is geometric in nature, relying on certain results from convex geometry, and it highlights that a negligible surface-to-volume ratio is indeed required for the validity of the approximation.

Next, \cref{sec:Fock} introduces the physical characteristic of a varying particle number by postulating the existence of a number operator with natural properties. We will show that this assumption implies the necessity of Fock space (\cref{pro:FockSpace}):
\begin{itemize}
    \item \itshape If there exists a total particle number operator with physically reasonable properties (pure point spectrum and $n$-particle eigenspaces), then the kinematical structure of the system is given in terms of Fock space.
\end{itemize}
We provide two proofs of this statement, one using only the spectral theorem and another one relying on the universal property of the Hilbert space direct sum. As discussed in detail in \cref{rem:FockSpaceAQFT} \ref{enu:remFockSpaceAQFT}, the above assertion should be seen as complementing similar results in the literature on the Fock space representation.

Finally, in \cref{sec:effectiveHamiltonian} we synthesize the previous results in order to derive the effective Hamiltonian given in \cref{eq:Heff} and to show that $H_\seff$ is, in the specific physical setting considered, unique up to a constant; this constitutes the main result of the paper:

\begin{theorem*}[Informal, see \cref{thm:effectiveHamiltonian}]
    If the surface-to-volume ratio approximation can be applied and the particle number is varying, then $H_\SFT \approx H_\seff \otimes \id_\SFR$ up to higher-order terms in the particle number of the open system. Moreover, within the first-order approximation this decomposition is unique up to additive constants.
\end{theorem*}

As a consequence of this theorem, we will show that the effective first-order approximation of the Hamiltonian $H_\SFT$ leads to a modified von Neumann equation that admits the grand canonical density $\rho_\mathrm{gc}$ as a natural stationary solution (\cref{cor:vonNeumannEquation}); thus, the theorem and the corollary together establish the two inferences represented by the solid arrows in \cref{fig:illustrationProblem}. We emphasize again that the effective Hamiltonian $H_\seff$ is a concept more fundamental than that of a statistical ensemble, with direct implications for physical systems \cite{Reible2025pre, Reible2025apq}, and that our chain of reasoning, therefore, offers a new perspective on grand canonical systems, where the effective Hamiltonian is usually taken for granted and a kinematical description in terms of Fock space is justified \emph{a posteriori} only.

\subsection{Notation}

To conclude the introduction, let us fix some notation, frequently employed throughout the text, which is not used uniformly in the literature.

For $d \in \N$ and a Borel subset $X \subset \R^d$, we will denote its $d$-dimensional Lebesgue measure by $\vol (X)$ or $\lambda^d (X)$, and its $(d-1)$-dimensional Hausdorff measure by $\area (X)$. Furthermore, $B^d = \set{x \in \R^d \, : \, \abs{x} < 1}$ and $S^{d-1} = \set{x \in \R^d \, : \, \abs{x} = 1}$ will denote the Euclidean unit ball and unit sphere, respectively, $B^d(a, r) = \set{x \in \R^d \, : \, \abs{x - a} < r}$ denotes the ball of radius $r > 0$ around $a \in \R^d$, and $\omega_d \ce \vol (B^d)$ is the volume of the unit ball.

General Hilbert spaces will be denoted by variants of the symbol $\MH$, and inner products are written as $\braket{\cdot, \cdot}$; the latter are taken to be linear in the second and semi-linear in the first argument. The symbol $\LO(\MH)$ denotes the set of all linear operators $T : \MH \supset \dom(T) \to \MH$ (generally unbounded) defined on a subspace $\dom(T)$ of $\MH$, $\BO(\MH) \subset \LO(\MH)$ will be the space of all everywhere-defined bounded linear operators on $\MH$, and finally $\BO_1(\MH) \subset \BO(\MH)$ denotes the space of all trace-class operators on $\MH$.

\section{Basic physical assumptions and mathematical setup}\label{sec:setup}

To set the stage, we introduce two basic hypotheses, \cref{asm:systemDivision,asm:interactionTwoBody}, which characterize the physical systems under investigation and thus provide the basis for the underlying physico-mathematical model; we also use the opportunity to fix some more notation. In \cref{subsec:densityOperators}, we establish a partial trace formula for unbounded operators, which is central for deriving our main results.

\subsection{Open quantum systems}

Consider a large, isotropic quantum system, referred to as the \emph{total system} $\mathsf{T}$, consisting of $N \in \N$ particles\footnote{With the term \enquote{particles}, we do not only mean elementary particles, but also larger composites that behave quantum-mechanically, e.g., molecules.} confined to a region $\Omega \subset \R^d$, $d \in \N$, of volume $V \ce \vol (\Omega) > 0$. Let $\MH_\SFT$ be the Hilbert space and $H_\SFT \in \LO(\MH_\SFT)$ be the self-adjoint Hamiltonian describing the physics of $\SFT$. The foundation for all the subsequent investigations is the following basic modeling assumption, which is the usual starting point in the theory of open quantum systems \cite{BreuerPetruccione2007, Davies1976, Vacchini2024}:

\begin{assumption}[Division of the total system]\label{asm:systemDivision}
    The total system $\SFT$ is divided into two interacting subsystems:
    \begin{enumerate}[label=(\arabic*)]
        \item the \emph{open system} $\SFS$, which contains $N_\SFS \in \N$ particles of interest and is confined to a bounded and convex subregion $\Omega_\SFS \subset \Omega$ of volume $V_\SFS \ce \vol (\Omega_\SFS) < + \infty$;
        \item the \emph{reservoir} $\SFR$, which is the remainder of $\SFT$ without $\SFS$, thus it contains $N_\SFR \ce N - N_\SFS$ particles confined to the region $\Omega_\SFR \ce \Omega \setminus \Omega_\SFS$ of volume $V_\SFR \ce \vol (\Omega_\SFR) > 0$.
    \end{enumerate}
    Regarding the physical scale of these two systems, we will assume that the open system is macroscopic, yet significantly smaller than the reservoir (and hence the total system), which is typical in statistical mechanics \cite[p. 149]{Huang1987}:
    \begin{enumerate}[labelindent=\parindent, leftmargin=*, label=\normalfont\textbf{A1.1.}, ref=A1.1]
        \item \label{enu:A1.1} $1 \ll N_\SFS \ll N_\SFR$ and $1 \ll V_\SFS \ll V_\SFR$.
    \end{enumerate}
    Finally, we will only be concerned with equilibrium (or near-equilibrium) situations and hence also pose the following assumption:
    \begin{enumerate}[labelindent=\parindent, leftmargin=*, label=\normalfont\textbf{A1.2.}, ref=A1.2]
        \item \label{enu:A1.2} The systems $\SFS$ and $\SFR$ have reached thermodynamic equilibrium.
    \end{enumerate}
\end{assumption}

\begin{remarks}
    \leavevmode
    \begin{enumerate}[env]
        \item Starting with the total system $\SFT$ and then dividing it into the subsystems $\SFS$ and $\SFR$, which have to be chosen according to certain physical requirements (like Assumption \ref{enu:A1.1} and others following below), facilitates a first-principle approach because no specific reservoir model is assumed. The opposite procedure, i.e., taking the open system and reservoir as given and then defining the total system, would lead to the same physics for the system $\SFS$, presupposing, however, information about the reservoir and would hence render the following investigations partially empirical.

        \item Assuming that the open system is confined to a bounded region is necessary in order to ensure that it is thermodynamically stable at positive temperature \cite[Rem. 3]{Hainzl2009II}; see also Refs.~\cite[p. 252]{LiebSeiringer2010} and \cite[Thm. XIII.76]{RS4}. Moreover, since $\SFS$ is assumed to be macroscopic and the total system is isotropic, the condition that $\Omega_\SFS$ should be convex, while restricting mathematical generality, is nevertheless rather natural from a physical point of view.
    \end{enumerate}
\end{remarks}

Corresponding to the division of the total system $\SFT$, we will assume, in accordance with the usual axioms of quantum mechanics \cite[Axiom A7, p. 844]{Moretti2017}, that $\MH_\SFT$ factorizes tensorially into a Hilbert space $\MH_\SFS$ of the open system and a Hilbert space $\MH_\SFR$ of the reservoir:
\begin{equation}\label{eq:totalHilbertSpace}
    \MH_\SFT = \MH_\SFS \otimes \MH_\SFR \ .
\end{equation}
Let $H_\SFS \in \LO(\MH_\SFS)$ and $H_\SFR \in \LO(\MH_\SFR)$ denote the self-adjoint Hamiltonians for the degrees of freedom of the open system and the reservoir, respectively. Building on \cref{eq:totalHilbertSpace}, the total Hamiltonian $H_\SFT$ can be written formally as \cite[p. 149]{Davies1976}
\begin{equation}\label{eq:totalHamiltonian}
    H_\SFT = H_\SFS \otimes \id_\SFR + \id_\SFS \otimes H_\SFR + H_\sint \ ,
\end{equation}
where the operator $H_\sint \in \LO(\MH_\SFS \otimes \MH_\SFR)$ mediates the interaction between $\SFS$ and $\SFR$. Since $H_0 \ce H_\SFS \otimes \id_\SFR + \id_\SFS \otimes H_\SFR$ is essentially self-adjoint on its canonical domain \cite[Thm. VIII.33]{RS1}, \cite[Thm. 7.23]{Schmuedgen2012}, the operator $H_\sint$ should be chosen such that $H_\SFT$ is essentially self-adjoint as well, e.g., as an appropriate perturbation of $H_0$ (see also \cref{exa:L2} below). The mathematical setup discussed so far is illustrated in \cref{fig:division}.

Looking now at the internal structure of the open system only, if $\MH_i$ denotes the Hilbert space of the $i$-th particle in $\SFS$, $i \in \set{1, \dotsc, N_\SFS}$, then similarly to \cref{eq:totalHilbertSpace} we have
\begin{equation}\label{eq:HilbertSpaceS}
    \MH_\SFS = \bigotimes_{i=1}^{N_\SFS} \MH_i \ .
\end{equation}
Let us write $h_i \in \LO(\MH_i)$ for the non-interacting self-adjoint Hamiltonian of the $i$-th particle in $\SFS$ such that $H_\SFS$ takes the form
\begin{equation}\label{eq:HamiltonianS}
    H_\SFS = \sum_{i=1}^{N_\SFS} h_i + H_{\SFS, \sint} \ ,
\end{equation}
where $H_{\SFS, \sint} \in \LO(\MH_\SFS)$ contains the interactions between particles in $\SFS$, and $\sum_{i=1}^{N_\SFS} h_i$ is an abbreviation for the more precise (but cumbersome) expression
\begin{equation*}
    \sum_{i=1}^{N_\SFS} \Bigl( \id_{\MH_1} \otimes \dotsm \otimes \id_{\MH_{i-1}} \otimes h_i \otimes \id_{\MH_{i+1}} \otimes \dotsm \otimes \id_{\MH_{N_\SFS}} \Bigr) \ .
\end{equation*}
To ensure stability of matter, we will assume, without loss of generality, that $h_i$ is a positive operator for all $i \in \set{1, \dotsc, N_\SFS}$; in this case, $H_{\SFS, 0} \ce \sum_{i=1}^{N_\SFS} h_i$ is self-adjoint and positive as well \cite[Thm. 3.8 (iv)]{Arai2024}, \cite[Thm. 7.23]{Schmuedgen2012}.

\begin{remark}\label{rem:indistinguishability}
    One should keep in mind that for indistinguishable particles, the tensor product in \cref{eq:HilbertSpaceS} has to be replaced by the symmetric or anti-symmetric tensor product of the Hilbert spaces $\MH_i$. Moreover, instead of simply assuming that there are $N_\SFS$ particles in $\SFS$, one has to consider all possible ways of taking $N_\SFS$ particles from the total system consisting of $N$ particles, which results in an additional combinatorial prefactor $\binom{N}{N_\SFS}$ to quantities like the density operator of the system \cite{DelleSite2024jpa}. In the following, we will not deal explicitly with these consequences of indistinguishability but leave the matter at this remark.
\end{remark}

\begin{figure}
    \centering
    \scalebox{0.55}{%
        \begin{tikzpicture}[font=\huge]
            % Shapes
            \draw[line width=2pt, fill=shadecolor] (0,0) ellipse (11 and 5);
            \draw[line width=2pt, fill=white, dashed] (-1.5,-1.5) rectangle (-7.5,1.5);
            % Total system
            \node at (0,7) {Total System $\SFT$};
            \node at (0,6) {$(\MH_\SFT, H_\SFT)$};
            \node at (0,4) {$\Omega$};
            \node at (0,-6) {$\Omega = \Omega_\SFS \cup \Omega_\SFR$, $N = N_\SFS + N_\SFR$};
            % Open system
            \node at (-4.5,0.5) {Open System $\SFS$};
            \node at (-4.5,-0.5) {$(N_\SFS, \MH_\SFS, H_\SFS)$};
            \node (OS) at (-4.5,-2.5) {$\Omega_\SFS$};
            \draw [->, very thick] (OS) to [out=180, in=270] (-6.75,-1.0);
            % Reservoir
            \node at (4.5,0.5) {Reservoir $\SFR$};
            \node at (4.5,-0.5) {$(N_\SFR, \MH_\SFR, H_\SFR)$};
            \node at (4.5,-2.5) {$\Omega_\SFR = \Omega \setminus \Omega_\SFS$};
            % Interaction
            \draw[<->, very thick] (-0.75,-0.5) -- (-2.25,-0.5);
            \node at (0,-0.5) {$H_\sint$};
        \end{tikzpicture}
    }
    \caption{Division of the total system according to \cref{asm:systemDivision}.}
    \label{fig:division}
\end{figure}

\subsection{Interactions}\label{subsec:interactions}

Coming to the interaction operators $H_\sint$ and $H_{\SFS, \sint}$, we will make the following simplifying assumption throughout the paper, which is typically sufficient for the needs of mathematical quantum mechanics \cite[p. 44]{LiebSeiringer2010}.

\begin{assumption}[Two-body interactions]\label{asm:interactionTwoBody}
    The interaction between particles within the open system as well as between particles from the open system and the reservoir is of two-body type, meaning that for all pairs $(i, j)$ and $(i, k)$ of particle indices $i, k \in \set{1, \dotsc, N_\SFS}$, $j \in \set{1, \dotsc, N_\SFR}$, $i \neq k$, there exist operators $V_{ij}$ acting in $\MH_\SFS \otimes \MH_\SFR$ and $W_{ik}$ acting in $\MH_\SFS$ such that 
    \begin{equation}\label{eq:interactionHamiltonians}
        H_\sint = \sum_{i=1}^{N_\SFS} \sum_{j=1}^{N_\SFR} V_{ij} \quad \text{and} \quad H_{\SFS, \sint} = \sum_{i=1}^{N_\SFS} \sum_{k > i}^{N_\SFS} W_{ik} \ .
    \end{equation}
    Moreover, we will assume that these operators can be chosen so as to render the Hamiltonians $H_\SFT = H_0 + H_\sint$ and $H_\SFS = H_{\SFS, 0} + H_{\SFS, \sint}$ (essentially) self-adjoint.
\end{assumption}

\begin{example}\label{exa:L2}
    In the standard mathematical model of non-relativistic quantum mechanics \cite[Axiom A5, p. 624]{Moretti2017}, the Hilbert space of a single particle from the open system is given by $L^2 (\Omega_\SFS, \lambda^d)$, hence \cite[Prop. 10.28]{Moretti2017}
    \begin{equation*}
        \MH_\SFS = \bigotimes\nolimits^{N_\SFS} L^2 (\Omega_\SFS, \lambda^d) \cong L^2 \bigl(\Omega_\SFS^{N_\SFS}, \lambda^{d N_\SFS}\bigr) \ .
    \end{equation*}
    The operators $h_i$ are formally given by $h_i = - \frac{\hbar^2}{2 m} \, \Delta_i + \phi(x_i)$, where $\Delta_i$ is the Laplacian of the $i$-th particle's coordinates (with suitable boundary conditions) and $\phi(x_i)$ is an external potential. If $U : \R^d \times \R^d \to \R$, $(x, y) \mapsto U(\abs{x - y})$, denotes the interaction potential, then one may first consider the multiplication operator $M_U$ acting in $L^2(\Omega_\SFS \times \Omega_\SFS, \lambda^{2d}) = L^2(\Omega_\SFS, \lambda^d) \otimes L^2(\Omega_\SFS, \lambda^d)$, and then construct $W_{ik} \in \LO(\MH_S)$ as follows:
    \begin{equation*}
        W_{ik} \ce \id_{\MH_1} \otimes \dotsm \otimes \id_{\MH_{i-1}} \otimes M_U \otimes \id_{\MH_{i+1}} \otimes \dotsm \otimes \id_{\MH_{k-1}} \otimes \id_{\MH_{k+1}} \otimes \dotsm \otimes \id_{\MH_N} \ .
    \end{equation*}
    Note that since $\bigotimes_{i=1}^{N_\SFS} \MH_i \cong \bigotimes_{i=1}^{N_\SFS} \MH_{\sigma(i)}$ for any permutation $\sigma$ of the set $\set{1, \dotsc, N}$, $W_{ik}$ is well-defined. An analogous procedure leads to the operator $V_{ij}$ appearing in $H_\sint$.
    
    In order to satisfy in this $L^2$-picture the requirement that $H_\SFS$ (and $H_\SFT$) are (essentially) self-adjoint operators, one should choose the potential $U$, for example, from an appropriate class of perturbations, namely, the Kato-Rellich class or the Rollnik class, in order to apply the Kato-Rellich or KLMN theorem; see, e.g., Refs.~\cite[Sects. X.2 and X.4]{RS2} and \cite[Sects. 8.3 and 10.7]{Schmuedgen2012} as well as Refs.~\cite{Simon1971paper, Simon1971book}.
\end{example}

\subsection{Density operators}\label{subsec:densityOperators}

Let $\DM(\MH)$ denote the set of density operators on the Hilbert space $\MH$, that is, the set of positive $\rho \in \BO_1(\MH)$ satisfying $\tr_\MH(\rho) = 1$. It is well-known \cite[Thm. 26.9]{BlanchardBruening2015} that $\rho \in \DM(\MH)$ if and only if there exist (i) a sequence $(\alpha_n)_{n \in \N} \subset [0,1]$ with $\sum_{n \in \N} \alpha_n = 1$ and (ii) an orthonormal basis $(e_n)_{n \in \N} \subset \MH$ (ONB for short) such that
\begin{equation*}
    \rho = \sum_{n=1}^{\infty} \alpha_n P_{e_n} \ ,
\end{equation*}
where $P_{e_n}$ denotes the projection onto the one-dimensional subspace spanned by the vector $e_n$, that is, the operator $P_{e_n} \xi \ce \braket{e_n, \xi} e_n$ ($\xi \in \MH$).

While for bounded operators $A \in \BO(\MH)$ the expression $\tr(A \rho)$ is always well-defined since $\BO_1(\MH)$ is a two-sided ideal in $\BO(\MH)$, this need not be the case for unbounded $T \in \LO(\MH)$ due to domain issues. In the literature, there are different necessary and sufficient conditions ensuring that $T \rho \in \BO_1(\MH)$; see, for example, Refs.~\cite[Prop. 11.27]{Moretti2017} and \cite[Prop. 4.59]{Moretti2019}. In this paper, we will use a more direct approach, inspired by Ref.~\cite[Eq. (3.1.26)]{LiebSeiringer2010}, to make sense of expressions of the form $\tr (T \rho)$ for unbounded $T$.

\begin{definition}[Compatible density operators]\label{def:compatibleDM}
    Let $T \in \LO(\MH)$ be self-adjoint. Define
    \begin{equation*}
        \setlength{\abovedisplayskip}{14pt}
        \setlength{\belowdisplayskip}{14pt}
        \DM(\MH, T) \ce \left\{ \rho \in \DM(\MH) \ : \ \parbox{24em}{\centering $\exists \, (e_n)_{n \in \N} \subset \dom(T) \ \exists \, (\alpha_n)_{n \in \N} \subset [0,1]:$ \\ (i) $(e_n)_{n \in \N}$ is an ONB for $\MH$ and $\sum_{n \in \N} \alpha_n = 1$; \\ (ii) $\rho = \sum_{n \in \N} \alpha_n P_{e_n}$; (iii) $\sum_{n \in \N} \alpha_n \abs{\braket{e_n, T e_n}} < + \infty$} \right\}
    \end{equation*}
    and call the elements of this set \emph{$T$-compatible density operators}. Furthermore, for any $\rho \in \DM(\MH, T)$ define the quantity $\tr(T \rho)$ by
    \begin{equation}\label{eq:trace}
        \tr(T \rho) \ce \sum_{n=1}^{\infty} \alpha_n \braket{e_n, T e_n} \ .
    \end{equation}
\end{definition}

Note that since the series on the right-hand side of \eqref{eq:trace} converges absolutely by assumption, $\tr(T \rho)$ is a real number, interpreted as the expectation value of $T$ in the mixed state $\rho$. We will, therefore, also use the notation
\begin{equation*}
    \EV_\rho [T] \ce \tr(T \rho) \ .
\end{equation*}
In \cref{app:densityOperators}, we show that our definition of $T$-compatible density operators is meaningful by establishing that the value of $\tr(T \rho)$ does not depend on the representation of $\rho \in \DM(\MH, T)$, and that for suitable $T$ this set contains non-trivial mixed states (\cref{lem:existenceCompatibleDM}); of course for every $T$ with $\dom(T) \neq \varnothing$ it holds that $\DM(\MH, T) \neq \varnothing$ since the pure states induced by vectors from the domain of $T$ are elements of this set.

Denote by $\tr_\SFR : \BO_1(\MH_\SFS \otimes \MH_\SFR) \to \BO_1(\MH_\SFS)$ the partial trace operation on $\MH_\SFT$ with respect to $\MH_\SFR$, that is, for every trace-class operator $\rho$ on $\MH_\SFS \otimes \MH_\SFR$, $\tr_\SFR (\rho)$ is the unique trace-class operator on $\MH_S$ which satisfies for all $A \in \BO(\MH_\SFS)$ \cite{Blanchard2011}, \cite[Cor. 26.5]{BlanchardBruening2015}:
\begin{equation*}
    \tr_{\MH_\SFS \otimes \MH_\SFR} \bigl( (A \otimes \id_\SFR) \rho \bigr) = \tr_{\MH_\SFS} \bigl( A \tr_\SFR (\rho) \bigr) \ .
\end{equation*}
Below, we will need this identity also for unbounded operators. With the help of \cref{def:compatibleDM}, the following result, whose proof we give in \cref{subsec:proofPT}, holds true.

\begin{lemma}\label{lem:partialTrace}
    Let $T \in \LO(\MH_\SFS)$ be self-adjoint and $\rho \in \DM(\MH_\SFS \otimes \MH_\SFR, T \otimes \id_\SFR)$ be a $T \otimes \id_\SFR$-compatible density operator. Then $\tr_\SFR (\rho) \in \DM(\MH_\SFS, T)$ and
    \begin{equation}\label{eq:partialTrace}
        \tr_{\MH_\SFS \otimes \MH_\SFR} \bigl( (T \otimes \id_\SFR) \rho \bigr) = \tr_{\MH_\SFS} \bigl( T \tr_\SFR (\rho) \bigr) \ .
    \end{equation}
\end{lemma}

In \cref{sec:STVR}, we will be interested in density operators $\rho$ on $\MH_\SFT$ for which the expected energies $\EV_{\rho_\SFS} [H_\SFS]$ of $\SFS$ and $\EV_{\rho_\SFR} [H_\SFR]$ of $\SFR$ as well as the expected interaction energy $\EV_\rho [H_\sint]$ are well-defined; therefore, we introduce the following set of \emph{admissible density operators}:\footnote{We note that the superscript \enquote{(1)} is added because in \cref{sec:effectiveHamiltonian}, we will consider a second class $\DM_\mathrm{adm}^{(2)} (\MH_\SFT)$ of admissible density operators.}
\begin{equation}\label{eq:admissibleDM}
    \DM_\mathrm{adm}^{(1)} (\MH_\SFT) \ce \DM(\MH_\SFT, H_\SFS \otimes \id_\SFR) \cap \DM(\MH_\SFT, \id_\SFS \otimes H_\SFR) \cap \DM(\MH_\SFT, H_\sint) \ .
\end{equation}
As a special case of \cref{lem:intersectionCompatibleDM}, we obtain the relation
\begin{equation}\label{eq:admissibleDMInclusion}
    \DM_\mathrm{adm}^{(1)} (\MH_\SFT) \subset \DM(\MH_\SFT, H_\SFT) \ .
\end{equation}

\section{Surface-to-volume ratio approximation}\label{sec:STVR}

When discussing the division of a large system into two subsystems in statistical mechanics, the interaction between these subsystems (that is, the operator $H_\sint$ in our case) is usually neglected based on the reason that it represents surface effects, which are negligible compared to the bulk energies of the subsystems; see, e.g., Refs.~\cite[pp. 131, 149]{Huang1987}, \cite[pp. 61, 78]{Nolting2018}, and \cite[pp. 38, 50]{Schwabl2006}. This is expressed formally as
\begin{equation}\label{eq:heuristicSTVR}
    H_\SFT \approx H_\SFS + H_\SFR \ ,
\end{equation}
with the comment that the operator $H_\sint$, while necessary for the equilibration of the two subsystems $\SFS$ and $\SFR$, is \enquote{negligibly small} compared to the operators $H_\SFS$ and $H_\SFR$ \cite[p. 50]{Schwabl2006}; \cref{eq:heuristicSTVR} is known as the \emph{surface-to-volume ratio approximation} \cite[p. 131]{Huang1987}.

Since the involved operators $H_\SFS$, $H_\SFR$, and $H_\sint$ are unbounded, in general, one cannot simply compare their norms to determine whether $H_\sint$ is small compared to $H_\SFS$ and $H_\SFR$. Therefore, in the following, we shall make the approximation \eqref{eq:heuristicSTVR} mathematically precise in a framework that can handle unbounded operators, and we will see in what specific sense $H_\sint$ can be neglected compared to $H_\SFS$ and $H_\SFR$. Our proof will also show that a negligibly small surface-to-volume ratio is indeed crucial for the validity of the approximation.

\subsection{Short-range interactions}

As a first step towards a rigorous version of \eqref{eq:heuristicSTVR}, we need to formulate an additional physical assumption on the system--reservoir interaction Hamiltonian $H_\sint$, which is one of two essential requirements for the validity of the approximation; in less precise terms, this assumption is also present in physical discussions of the surface-to-volume ratio approximation \cite[p. 131]{Huang1987}.

\begin{assumption}[Interaction cutoff]\label{asm:interactionCutoff}
    There exists a number $\delta > 0$, which represents the \emph{effective interaction distance} between particles from $\SFS$ and $\SFR$, such that the operator $H_\sint$ is determined only by those particles from $\Omega_\SFS$ and $\Omega_\SFR$ whose distance to each other is smaller than $\delta$.
\end{assumption}

\begin{remarks}
    \leavevmode
    \begin{enumerate}[env]
        \item \cref{asm:interactionCutoff} is natural in that for systems which are dominated by long-range interactions between the two subsystems $\SFS$ and $\SFR$, one cannot expect to neglect the operator $H_\mathrm{int}$ compared to the other terms of the total Hamiltonian \eqref{eq:totalHamiltonian}. In this regard, it is interesting to observe that the Coulomb interaction, which is the most important interaction for ordinary matter but also characterized by a slow decay at infinity, can be considered to have a small effective interaction distance due to electrostatic screening which is a consequence of Newton's theorem \cite{Lieb1972, LiebSeiringer2010} (see also Refs.~\cite{Hainzl2009I, Hainzl2009II}); this fact is indeed essential for the existence of the thermodynamic limit of Coulomb systems. In applied physics, the effective interaction distance of the Coulomb potential can often be characterized in terms of the Debye length; see Ref.~\cite{Brydges1999} and references therein for the range of validity of different screening models.

        \item For specific applications of the surface-to-volume ratio approximation, one must check whether the model under investigation satisfies \cref{asm:interactionCutoff} by determining the effective interaction distance of the model. As we will show below, the error terms due to the approximation depend on $\delta$, and thus, once the latter is known, one can estimate the correction terms and decide whether the approximation is feasible or not. In other words, while in the present work the parameter $\delta$ is a purely mathematical tool, in general it carries importance for practical applications of our results.
    \end{enumerate}
\end{remarks}

\cref{asm:interactionCutoff} implies that for a particle at position $x \in \Omega_\SFS$, only those particles contained within the ball $B(x, \delta)$ of radius $\delta$ around $x$ will interact with it. Regarding inter-domain interactions between $\Omega_\SFS$ and $\Omega_\SFR$, this means that all the involved particles necessarily must be located in the tubular $\delta$-neighborhood $(\partial \Omega_\SFS)_{\delta}$ (cf. Definition \ref{def:epsilonNeighborhood} in \cref{app:convexGeometry}) of the boundary $\partial \Omega_\SFS$ of the region $\Omega_\SFS$:
\begin{equation*}
    (\partial \Omega_\SFS)_{\delta} = \bigcup_{z \in \partial \Omega_\SFS} B(z, \delta) = \set[\big]{y \in \Omega \, : \, \dist(y, \partial \Omega_\SFS) < \delta} \ .
\end{equation*}
Note that since $\Omega = \Omega_\SFS \cup \Omega_\SFR$ by construction (\cref{asm:systemDivision}), we may split $(\partial \Omega_\SFS)_{\delta}$ into an \enquote{outer} tubular neighborhood $(\partial \Omega_\SFS)_{\delta}^+$ extending into $\Omega_\SFR$, and an \enquote{inner} tubular neighborhood $(\partial \Omega_\SFS)_{\delta}^-$ extending into $\Omega_\SFS$, which are given by
\begin{align*}
    (\partial \Omega_\SFS)_{\delta}^+ &\ce \set[\big]{y \in \Omega_\SFR \, : \, \dist(y, \partial \Omega_\SFS) < \delta} \ , \\
    (\partial \Omega_\SFS)_{\delta}^- &\ce \set[\big]{x \in \Omega_\SFS \, : \, \dist(x, \partial \Omega_\SFS) < \delta} \ .
\end{align*}
It clearly holds that $(\partial \Omega_\SFS)_{\delta} = (\partial \Omega_\SFS)_{\delta}^+ \cup (\partial \Omega_\SFS)_{\delta}^-$; see also \cref{fig:corridor} for an illustration of the different involved sets. To emphasize that this tubular neighborhood has total width $2 \delta$, we shall denote it by the symbol
\begin{equation*}
    \Gamma_{2 \delta} (\Omega_\SFS) \ce (\partial \Omega_\SFS)_{\delta}^+ \cup (\partial \Omega_\SFS)_{\delta}^- = (\partial \Omega_\SFS)_\delta \ .
\end{equation*}
In this terminology, \cref{asm:interactionCutoff} entails that all the particles from $\Omega_\SFS$ and $\Omega_\SFR$ captured by $H_\sint$ lie inside the \emph{interaction corridor} $\Gamma_{2 \delta } (\Omega_\SFS)$. Let $N_\sint \equiv N_\sint (\delta) \in \N$ be the number of these particles. Then we may rewrite the interaction Hamiltonian from \cref{eq:interactionHamiltonians} as
\begin{equation}\label{eq:Hint}
    H_\sint = \sum_{i=1}^{N_\sint} \sum_{j>i}^{N_\sint} V_{ij}  \ .
\end{equation}

\begin{figure}
    \centering
    \scalebox{0.6}{%
        \begin{tikzpicture}[font=\huge]
        % Definitions
        \def\offsetDist{1.25cm}
        % Coordinates
        \coordinate (A) at (0,-2);
        \coordinate (B) at (5,-0.5);
        \coordinate (C) at (7,2.5);
        \coordinate (D) at (5,6);
        \coordinate (E) at (0,6.5);
        \coordinate (F) at (-5,5);
        \coordinate (G) at (-7,1.5);
        \coordinate (H) at (-5,-1.5);
        % Definition of path
        \draw[very thick, save path=\omegaone] plot [smooth cycle, tension=0.75] coordinates {(A) (B) (C) (D) (E) (F) (G) (H)};
        % Corridor around path
        \pgfoffsetpath{\omegaone}{-\offsetDist}
        \pgfsetfillcolor{shadecolor}\pgfusepath{fill}
        \pgfoffsetpath{\omegaone}{-\offsetDist}
        \pgfsetdash{{5pt}{4pt}}{0pt} \pgfsetlinewidth{1pt} \color{black} \pgfusepathqstroke
        \pgfoffsetpath{\omegaone}{\offsetDist}
        \pgfsetfillcolor{white}\pgfusepath{fill}
        \pgfoffsetpath{\omegaone}{\offsetDist}
        \pgfsetdash{{5pt}{4pt}}{0pt} \pgfsetlinewidth{1pt} \color{black} \pgfusepathqstroke
        \pgfsetdash{}{0pt}
        % Omega_S
        \draw[line width=2pt] plot [smooth cycle, tension=0.75] coordinates {(A) (B) (C) (D) (E) (F) (G) (H)};
        % Labels for sets
        \node at (-7.5,7) {$\Omega_\SFR$};
        \node at (0,2.25) {$\Omega_\SFS$};
        \node (boundary) [below, xshift=3.5cm, yshift=-1cm] at (B) {$\partial \Omega_\SFS$};
        \node (innerCorr) [left, xshift=-1.75cm, yshift=0.5cm] at (C) {$(\partial \Omega_\SFS)_{\delta}^-$};
        \node (outerCorr) [above right, xshift=1.5cm, yshift=1cm] at (D) {$(\partial \Omega_\SFS)_{\delta}^+$};
        \node (R0) [below left, xshift=-1.75cm, yshift=-1cm] at (G) {$2 \delta$};
    	% Arrows
    	\draw [->, very thick] (boundary) to [out=180, in=300] (B);
        \draw [->, very thick] (innerCorr) to [out=90, in=210] ([yshift=-0.75cm]D);
        \draw [->, very thick] (outerCorr) to [out=180, in=60] ([yshift=0.75cm]D);
        \draw [->, very thick] (R0) to [out=0, in=270] ([xshift=-0.5cm,yshift=-0.25cm]G);
        \draw [<->, line width=2pt, dotted] (-8.25,1.5) -- (-5.75,1.5);
    \end{tikzpicture}
    }
    \caption{Illustration of the set $\Omega_\SFS$, circumscribed by the solid line, and the interaction corridor $(\partial \Omega_\SFS)_{\delta}^+ \cup (\partial \Omega_\SFS)_{\delta}^-$ of width $2 \delta$ around its boundary, highlighted as the gray region delimited by the dashed lines.}
    \label{fig:corridor}
\end{figure}

\subsection{Volume of the interaction corridor}

The main technical ingredient for the surface-to-volume ratio approximation is the following lemma, which provides an estimate for the volume of the interaction corridor $\Gamma_{2 \delta} (\Omega_\SFS)$ in terms of the surface area of the boundary $\partial \Omega_\SFS$ and the interaction distance $\delta$; its proof relies on established results in convex geometry that are summarized in \cref{app:convexGeometry}.

\begin{lemma}\label{lem:volumeEstimate}
    Assume that $K \subset \R^d$ is a bounded convex domain. Then the volume of the corridor $\Gamma_{2 \delta} (K)$ is bounded above and below by
    \begin{equation}\label{eq:volumeEstimate}
        \delta \cdot \area (\partial K) + \MO (\delta^2) \le \vol \bigl( \Gamma_{2 \delta} (K) \bigr) \le 2 \delta \cdot \area (\partial K) + \MO (\delta^2) \quad \text{for} \quad \delta \to 0 \ .
    \end{equation}
\end{lemma}

\begin{proof}
    First, note that using Definition \ref{def:epsilonNeighborhood} of the tubular $\delta$-neighborhood $K_\delta$ of the set $K$, we may express $(\partial K)_\delta^+$ equivalently as $(\partial K)_\delta^+ = K_\delta \setminus K$. Since the outer and inner tubular neighborhoods of $K$ are disjoint, it follows from additivity of the Lebesgue measure that
    \begin{equation}\label{eq:volumeSplit}
        \vol \bigl( \Gamma_{2 \delta} (K) \bigr) = \vol \bigl( (\partial K)_\delta^+ \cup (\partial K)_\delta^- \bigr) = \vol \bigl( K_\delta \setminus K \bigr) + \vol \bigl( (\partial K)_\delta^- \bigr) \ .
    \end{equation}
    We shall consider both of the terms separately. For the first one, observe that since $K \subset K_\delta$ and $\vol (K) < + \infty$, we have $\vol (K_\delta \setminus K) = \vol (K_\delta) - \vol (K)$. Thus, we can apply the Minkowski-Steiner formula (\cref{thm:MinkowskiSteiner} in \cref{app:convexGeometry}) to obtain
    \begin{equation}\label{eq:auxiliaryVolumeEstimate}
        \vol \bigl( K_\delta \setminus K \bigr) = \area (\partial K) \, \delta + \MO(\delta^2) \quad \text{for} \quad \delta \to 0 \ .
    \end{equation}
    For the second term, note that by definition, one has $(\partial K)_\delta^- \subset \set{x \in K \, : \, \dist(x, \partial K) \le \delta}$. Therefore, monotonicity of the Lebesgue measure and \cref{lem:boundThinning} imply
    \begin{equation*}
        \vol \bigl( (\partial K)_\delta^- \bigr) \le \vol \bigl( \set{x \in K \, : \, \dist(x, \partial K) \le \delta} \bigr) \le \area (\partial K) \, \delta \ .
    \end{equation*}
    Combining the last three equations proves the upper bound. For the lower bound, it suffices to note that $\vol (\Gamma_{2 \delta} (K)) \ge \vol (K_\delta \setminus K)$ by \cref{eq:volumeSplit}, hence the estimate follows from \eqref{eq:auxiliaryVolumeEstimate}.
\end{proof}

\begin{remark}\label{rem:refinedVolumeEstimate}
    Let $K \subset \R^d$ satisfy the assumptions of \cref{lem:volumeEstimate} and suppose, in addition, that $B^d \subset K$. Then the quermassintegrals $W_j (K)$, $j \in \set{2, \dotsc, d-1}$, from \cref{eq:quermassintegral} entering the Minkowski-Steiner formula can be estimated using the fact that the mixed volume is non-decreasing in each argument (cf. \cref{rem:MinkowskiSteiner} in \cref{app:convexGeometry} and Ref.~\cite[Thm. 6.9]{Gruber2007}): for all $j \in \set{2, \dotsc, d-1}$ we have the upper bound
    \begin{equation*}
        W_j (K) = V \bigl( \, \underbrace{K, \dotsc, K}_{d - j}, \, \underbrace{B^d, \dotsc, B^d}_{j} \, \bigr) \le V \bigl( \, \underbrace{K, \dotsc, K}_{d - 1}, B^d \, \bigr) = W_1 (K)
    \end{equation*}
    as well as the lower bound
    \begin{equation*}
        \omega_d = V \bigl( B^d, \dotsc, B^d \bigr) \le V \bigl( \, \underbrace{K, \dotsc, K}_{d - j}, \, \underbrace{B^d, \dotsc, B^d}_{j} \, \bigr) = W_j (K) \ .
    \end{equation*}
    Since $W_1 (K) = \frac{1}{d} \, \area (\partial K)$, it follows that instead of using \cref{eq:auxiliaryVolumeEstimate} in the above proof, we can apply the previous two estimates for $W_j (K)$ in the Minkowski-Steiner formula \eqref{eq:MinkowskiSteiner} to obtain the upper bound
    \begin{align*}
        \vol \bigl( K_\delta \setminus K \bigr) &\le \area (\partial K) \, \delta + \sum_{j=2}^{d} \frac{1}{d} \binom{d}{j} \, \area (\partial K) \, \delta^j \\
        &= \area (\partial K) \, \delta + \MO \bigl( \area(\partial K) \, \delta^2 \bigr) \quad \text{for} \quad \delta \to 0 \ ,
    \end{align*}
    where we used $\omega_d \le W_1(K)$ for the last term, and, similarly, the lower bound
    \begin{equation*}
        \vol \bigl( K_\delta \setminus K \bigr) \ge \area (\partial K) \, \delta + \MO \bigl( \omega_d \, \delta^2 \bigr) \quad \text{for} \quad \delta \to 0 \ .
    \end{equation*}
    These two estimates for $\vol \bigl( K_\delta \setminus K \bigr)$ translate directly into a correspondingly modified version of the two-sided inequality \eqref{eq:volumeEstimate}: for $\delta \to 0$ we have
    \begin{equation}\label{eq:refinedVolumeEstimate}
        \delta \cdot \area (\partial K) + \MO \bigl( \omega_d \, \delta^2 \bigr) \le \vol \bigl( \Gamma_{2 \delta} (K) \bigr) \le 2 \delta \cdot \area (\partial K) + \MO \bigl( \area(\partial K) \, \delta^2 \bigr) \ .
    \end{equation}
\end{remark}

\subsection{Formulation of the approximation}

We can now formulate and prove the main result of this section, namely, the validity of the surface-to-volume ratio approximation under certain assumptions. As the name suggests, the quantity
\begin{equation*}
    \kappa_\SFS \ce \frac{\area (\partial \Omega_\SFS)}{V_\SFS}
\end{equation*}
will play an important role. Beyond the basic modeling Assumptions \ref{asm:systemDivision} and \ref{asm:interactionTwoBody}, and the effective interaction distance defined in \ref{asm:interactionCutoff}, the following proposition requires an additional physical hypothesis concerning $\kappa_\SFS$ and a compatibility condition on the density operator.

\begin{proposition}[Surface-to-volume ratio approximation]\label{pro:STVR}
    Suppose that Assumptions \ref{asm:systemDivision}, \ref{asm:interactionTwoBody}, and \ref{asm:interactionCutoff} hold true---recall that this entails that $\Omega_\SFS$ is bounded and convex---and assume, in addition, that $\Omega_\SFS$ is large enough such that
    \begin{enumerate}[labelindent=\parindent, leftmargin=*, label=\normalfont\textbf{A4.}, ref=A4]
        \item \label{enu:A4} $\kappa_\SFS \cdot \delta \ll 1$ and $B^d \subset \Omega_\SFS \ .$
    \end{enumerate}   
    Let $\rho \in \DM_\mathrm{adm}^{(1)} (\MH_\SFT)$ be an admissible density operator of the total system {\normalfont[}cf. \cref{eq:admissibleDM}{\normalfont]}. Assume additionally that $\rho$ satisfies
    \begin{enumerate}[labelindent=\parindent, leftmargin=*, label=\normalfont\textbf{A5.}, ref=A5]
        \item \label{enu:A5} $\displaystyle \rho \in \bigcap_{1 \le i < j \le N_\sint} \DM(\MH_\SFT, V_{ij})$ and $\displaystyle \rho_\SFS \in \bigcap_{1 \le i < j \le N_\SFS} \DM(\MH_\SFS, W_{ij}) \ ,$
    \end{enumerate}
    where $V_{ij}$ and $W_{ij}$ are the system--reservoir and system--system interaction operators introduced in \cref{asm:interactionTwoBody}, and $\rho_\SFS = \tr_\SFR (\rho)$ is the partial trace of $\rho$ with respect to the reservoir. Under these conditions,
    \begin{equation}\label{eq:STVR}
        \EV_\rho [H_\SFT] = \EV_\rho [H_\SFS \otimes \id_\SFR + \id_\SFS \otimes H_\SFR] + \MO \bigl( \kappa_\SFS^2 \cdot \delta^2 \bigr) \quad \text{for} \quad \delta \to 0 \ .
    \end{equation}
\end{proposition}

Before we give the proof of this proposition, we shall comment on the two new assumptions \ref{enu:A4} and \ref{enu:A5} as well as on the interpretation of the result \eqref{eq:STVR}.

\begin{remarks}
    \leavevmode
    \begin{enumerate}[env]
        \item The first condition $\kappa_\SFS \cdot \delta \ll 1$ in Assumption \ref{enu:A4} is a geometric formalization of the intuitive physical idea that the interaction between $\SFS$ and $\SFR$ represents negligible surface effects; this becomes particularly clear when $\Omega_\SFS$ is a ball (see \cref{exa:concreteSTVR} below). The second demand $B^d \subset \Omega_\SFS$ is not strictly necessary for the proof, but it allows us to apply the more controlled estimate \eqref{eq:refinedVolumeEstimate} of \cref{rem:refinedVolumeEstimate}---note that using \cref{eq:volumeEstimate} from \cref{lem:volumeEstimate} instead, one would obtain the same qualitative error $\MO (\kappa_\SFS^2 \cdot \delta^2)$. Assumption \ref{enu:A5} entails the physically very natural properties of having well-defined expected interaction energies in the states $\rho$ and $\rho_\SFS$. This ensures that the computations involving $\EV_\rho [H_\sint]$ and $\EV_{\rho_\SFS} [H_{\SFS, \sint}]$ in the following proof are mathematically justified.

        \item The key assertion of \cref{pro:STVR} is that we can \emph{approximate} the difficult-to-handle expression $\mathbb{E}_\rho [H_\mathsf{T}]$ by the more accessible expression $\mathbb{E}_\rho [H_\mathsf{S} \otimes \id_\mathsf{R} + \id_\mathsf{S} \otimes H_\mathsf{R}]$ \emph{up to an error} of second and higher orders in the quantity $\kappa_\SFS \cdot \delta$, which is a purely physical characteristic of our model. In particular, for systems satisfying Assumption \ref{enu:A4}, the error terms are negligibly small, and hence \cref{eq:STVR} justifies neglecting them:
        \begin{equation*}
            \EV_\rho [H_\SFT] \approx \EV_\rho [H_\SFS \otimes \id_\SFR + \id_\SFS \otimes H_\SFR] \ .
        \end{equation*}
        When using this approximation in the sequel, one should keep in mind that we have control over the error, and that we could, in principle, carry the error terms through, but \emph{within our physical model specified by \ref{enu:A4}}, that is, in a \emph{first-order approximation}, they are not relevant. This already constitutes an important step forward since in physics the approximation is used without any comment about the error; in this regard, we mention that the heuristic approximation \eqref{eq:heuristicSTVR} is applied when the energy of $\SFT$ is compared with that of $\SFS$ \cite[pp. 50 f.]{Schwabl2006}, and hence our result can indeed be considered a useful formalization of the former.
    \end{enumerate}
\end{remarks}

Let us now prove the previous proposition.

\begin{proof}[Proof of \cref{pro:STVR}]
    Due to the assumption $\rho \in \DM_\mathrm{adm}^{(1)} (\MH_\SFT)$, the quantities $\EV_\rho [H_\SFS \otimes \id_\SFR]$, $\EV_\rho [\id_\SFS\otimes H_\SFR]$, and $\EV_\rho [H_\sint]$ are all well-defined and finite. Moreover, since $\DM_\mathrm{adm}^{(1)} (\MH_\SFT) \subset \DM(\MH_\SFT, H_\SFT)$ according to \cref{eq:admissibleDMInclusion}, also $\EV_\rho [H_\SFT]$ is finite, and we have
    \begin{equation}\label{eq:expectedTotalEnergy}
        \EV_\rho [H_\SFT] = \EV_\rho [H_\SFS \otimes \id_\SFR] + \EV_\rho [\id_\SFS\otimes H_\SFR] + \EV_\rho [H_\sint] \ .
    \end{equation}
    Using the partial trace formula \eqref{eq:partialTrace} from \cref{lem:partialTrace}, we can write
    \begin{equation*}
        \EV_\rho [H_\SFS \otimes \id_\SFR] = \tr_{\MH_\SFS \otimes \MH_\SFR} \bigl((H_\SFS \otimes \id_\SFR) \rho\bigr) = \tr_{\MH_\SFS} \bigl(H_S \tr_\SFR (\rho)\bigr) = \EV_{\rho_\SFS} [H_\SFS] \ .
    \end{equation*}
    As specified in \cref{eq:HamiltonianS}, the open system's Hamiltonian is of the form $H_\SFS = H_{\SFS, 0} + H_{\SFS, \sint}$ with $H_{\SFS, 0} \ge 0$. By Assumption \ref{enu:A5} and \cref{lem:intersectionCompatibleDM}, we have $\rho_\SFS \in \DM(\MH_\SFS, H_{\SFS, \sint})$ and hence also $\rho_\SFS \in \DM(\MH_\SFS, H_{\SFS, 0})$. Let $E_{\SFS, 0} \ce \EV_{\rho_\SFS} [H_{\SFS, 0}]$. For the following argument, it proves useful to consider the shifted Hamiltonian
    \begin{equation}\label{eq:energyShift}
        \wt{H_\SFS} \ce H_\SFS - E_{\SFS, 0} \id_{\MH_\SFS} \quad \text{with} \quad \EV_{\rho_\SFS} [\wt{H_\SFS}] = \EV_{\rho_\SFS} [H_{\SFS, \sint}] \ .
    \end{equation}
    Since in the end, one can always reverse the shift in order to return to the original energy scale, we will simply assume that $H_\SFS$ has been shifted as to guarantee the previous identity without using a new symbol. The quantity $\EV_{\rho_\SFS} [H_{\SFS, \sint}]$ can then be bounded from below as follows (using implicitly \cref{lem:intersectionCompatibleDM}):
    \begin{align*}
        \EV_{\rho_\SFS} [H_{\SFS, \sint}] &= \tr_{\MH_\SFS} \Biggl(\Biggl[\,\sum_{i=1}^{N_\SFS} \sum_{j>i}^{N_\SFS} W_{ij}\Biggr] \rho_\SFS\Biggr) = \sum_{i=1}^{N_\SFS} \sum_{j>i}^{N_\SFS} \tr_{\MH_\SFS} (W_{ij} \rho_\SFS) \\
        &\ge \sum_{i=1}^{N_\SFS} \sum_{j>i}^{N_\SFS} \underbrace{\biggl[\min_{1 \le k < \ell \le N_\SFS} \tr_{\MH_\SFS} (W_{k \ell} \rho_\SFS)\biggr]}_{\ec E_{\SFS, \sint}^<} = N_\SFS (N_\SFS - 1) E_{\SFS, \sint}^< \approx N_\SFS^2 E_{\SFS, \sint}^< \ ,
    \end{align*}
    where the last approximation is valid because of $N_\SFS \gg 1$ by Assumption \ref{enu:A1.1}. In an analogous way, we can estimate $\EV_{\rho_\SFS} [H_{\SFS, \sint}]$ from above:
    \begin{equation*}
        \EV_{\rho_\SFS} [H_{\SFS, \sint}] \le \sum_{i=1}^{N_\SFS} \sum_{j>i}^{N_\SFS} \underbrace{\biggl[\max_{1 \le k < \ell \le N_\SFS} \tr_{\MH_\SFS} (W_{k \ell} \rho_\SFS)\biggr]}_{\ec E_{\SFS, \sint}^>} = N_\SFS (N_\SFS - 1) E_{\SFS, \sint}^> \approx N_\SFS^2 E_{\SFS, \sint}^> \ .
    \end{equation*}
    Similarly, we can estimate $\EV_\rho [H_\sint]$ from above and below: using \cref{eq:Hint}, Assumption \ref{enu:A5}, and \cref{lem:intersectionCompatibleDM}, it follows that
    \begin{align*}
        \EV_\rho [H_\sint] &= \tr_{\MH_\SFT} \Biggl( \Biggl[ \sum_{i=1}^{N_\sint} \sum_{j>i}^{N_\sint} V_{ij} \Biggr] \rho \Biggr) = \sum_{i=1}^{N_\sint} \sum_{j>i}^{N_\sint} \tr_{\MH_\SFT} (V_{ij} \rho) \\
        &\le \sum_{i=1}^{N_\sint} \sum_{j>i}^{N_\sint} \underbrace{\left[ \, \max_{1 \le k < \ell \le N_\sint} \tr_{\MH_\SFT} (V_{k \ell} \rho) \right]}_{\ec E_\sint^>} = N_\sint (N_\sint - 1) E_\sint^> \approx N_\sint^2 E_\sint^>
    \end{align*}
    as well as
    \begin{equation*}
        \EV_\rho [H_\sint] \ge \sum_{i=1}^{N_\sint} \sum_{j>i}^{N_\sint} \underbrace{\left[ \, \min_{1 \le k < \ell \le N_\sint} \tr_{\MH_\SFT} (V_{k \ell} \rho) \right]}_{\ec E_\sint^<} = N_\sint (N_\sint - 1) E_\sint^< \approx N_\sint^2 E_\sint^< \ .
    \end{equation*}
    Observe that by Assumption \ref{enu:A5}, the quantities $E_{\SFS, \sint}^<, E_{\SFS, \sint}^>, E_\sint^<, E_\sint^>$ are all well-defined and finite. With the previous four bounds (as well as our assumption below \cref{eq:energyShift}, which entails that $\EV_\rho [H_\SFS \otimes \id_\SFR] = \EV_{\rho_\SFS} [H_{\SFS, \sint}]$), we arrive at
    \begin{equation*}
        \frac{N_\sint^2}{N_\SFS^2} \frac{E_\sint^<}{E_{\SFS, \sint}^>} \le \frac{\EV_\rho [H_\sint]}{\EV_\rho [H_\SFS \otimes \id_\SFR]} \le \frac{N_\sint^2}{N_\SFS^2} \frac{E_\sint^>}{E_{\SFS, \sint}^<} \ .
    \end{equation*}
    
    For small enough $\delta$, we can assume that the particle number density $n_\sint$ inside the interaction corridor $\Gamma_{2 \delta} (\Omega_\SFS)$ is constant, i.e., $n_\sint = N_\sint / \vol \bigl(\Gamma_{2 \delta} (\Omega_\SFS)\bigr)$. This assumption is meaningful because $\SFS$ and $\SFR$ are in thermodynamic equilibrium (cf. Assumption \ref{enu:A1.2}), implying that the two parts $(\partial \Omega_\SFS)_{\delta}^- \subset \Omega_\SFS$ and $(\partial \Omega_\SFS)_{\delta}^+ \subset \Omega_\SFR$ of the interaction corridor must have the same particle number density. Therefore, we can recast the previous equation as
    \begin{equation*}
        \underbrace{\frac{n_\sint^2}{n_\SFS^2} \frac{E_\sint^<}{E_{\SFS, \mathrm{int}}^>}}_{\ec C_1 (n_\SFS)} \left( \frac{\vol \bigl(\Gamma_{2 \delta} (\Omega_\SFS)\bigr)}{V_\SFS} \right)^2 \le \frac{\EV_\rho [H_\sint]}{\EV_\rho [H_\SFS \otimes \id_\SFR]} \le \left( \frac{\vol \bigl(\Gamma_{2 \delta} (\Omega_\SFS)\bigr)}{V_\SFS} \right)^2 \underbrace{\frac{n_\sint^2}{n_\SFS^2} \frac{E_\sint^>}{E_{\SFS, \mathrm{int}}^<}}_{\ec C_2 (n_\SFS)} \ ,
    \end{equation*}
    where $n_\SFS = N_\SFS / V_\SFS$ is the uniform particle number density of the open system $\SFS$. The constants $C_1 (n_\SFS), C_2 (n_\SFS) \in \R$ depend on the physics of the open system $\SFS$, in particular, on its particle number density which we have highlighted as a functional dependency for later use, as well as on the nature of the interaction between $\SFS$ and the reservoir $\SFR$. From the refined volume estimate \eqref{eq:refinedVolumeEstimate} of \cref{rem:refinedVolumeEstimate},
    \begin{equation}\label{eq:ratioEstimate}
        \delta \cdot \frac{\area (\partial \Omega_\SFS)}{V_\SFS} + \MO \left( \frac{\omega_d}{V_\SFS} \cdot \delta^2 \right) \le \frac{\vol \bigl( \Gamma_{2 \delta} (\Omega_\SFS) \bigr)}{V_\SFS} \le 2 \delta \cdot \frac{\area (\partial \Omega_\SFS)}{V_\SFS} + \MO \left( \frac{\area (\partial \Omega_\SFS)}{V_\SFS} \cdot \delta^2 \right) \ ,
    \end{equation}
    which is applicable since $B^d \subset \Omega_\SFS$ according to the second demand of \ref{enu:A4}, we obtain the following bounds:
    \begin{equation}\label{eq:energyEstimate}
        C_1 (n_\SFS) \left[ (\kappa_\SFS \cdot \delta)^2 + \MO \biggl( \frac{\omega_d^2}{V_\SFS^2} \cdot \delta^4 \biggr) \right] \le \frac{\EV_\rho [H_\sint]}{\EV_\rho [H_\SFS \otimes \id_\SFR]} \le C_2 (n_\SFS) \left[ 4 \, (\kappa_\SFS \cdot \delta)^2 + \MO \bigl(\kappa_\SFS^2 \cdot \delta^4\bigr) \right] \ .
    \end{equation}
    Notice that our Assumption \ref{enu:A4} implies that already the first terms $(\kappa_\SFS \cdot \delta)^2$ on both sides will be very small; moreover, $\omega_d^2 / V_\SFS^2 \ll 1$ because $V_\SFS \gg 1$ by \ref{enu:A1.1}. Therefore, the previous two-sided estimates imply that
    \begin{equation*}
        \frac{\EV_\rho [H_\sint]}{\EV_\rho [H_\SFS \otimes \id_\SFR]} = \MO \bigl( \kappa_\SFS^2 \cdot \delta^2 \bigr) \quad \text{for} \quad \delta \to 0 \ .
    \end{equation*}
    We used the fact that constants can be neglected from big O notation, and we only kept the system constant $\kappa_\SFS$ as it is explicitly featured in \ref{enu:A4}. Returning now to the formula \eqref{eq:expectedTotalEnergy} for the expected energy of the total system, it follows that
    \begin{align*}
        \EV_\rho [H_\SFT] &= \EV_\rho [H_\SFS \otimes \id_\SFR] \left( 1 + \frac{\EV_\rho [H_\sint]}{\EV_\rho [H_\SFS \otimes \id_\SFR]} \right) + \EV_\rho [\id_\SFS \otimes H_\SFR] \\[2pt]
        &= \EV_\rho [H_\SFS \otimes \id_\SFR] \left( 1 + \MO \bigl( \kappa_\SFS^2 \cdot \delta^2 \bigr) \right) + \EV_\rho [\id_\SFS \otimes H_\SFR] \\[8pt]
        &= \EV_\rho [H_\SFS \otimes \id_\SFR + \id_\SFS\otimes H_\SFR] + \MO \bigl( \kappa_\SFS^2 \cdot \delta^2 \bigr) \ . \tag*{\qedhere}
    \end{align*}
\end{proof}

In passing, we mention that we had already linked the surface-to-volume ratio approximation to the so-called two-sided Bogoliubov inequality in the past, which provides upper and lower bounds for the free energy cost of separating a large system into independent subsystems \cite{Reible2022, Reible2025apx, Reible2023}; in this context, \cref{pro:STVR} provides a new geometric instead of the previous thermodynamic quantification of surface effects.

To conclude this section, we will consider two special cases ($d = 3$ and $\Omega_\SFS = B^d$) in which the error terms of the approximation \eqref{eq:STVR} can be determined more explicitly.

\begin{examples}\label{exa:concreteSTVR}
    \leavevmode
    \begin{enumerate}[env]
        \item Suppose that $d = 3$. In this case, the Minkowski-Steiner formula \eqref{eq:MinkowskiSteiner} for the volume of $(\partial \Omega_\SFS)_\delta^+ = (\Omega_\SFS)_\delta \setminus \Omega_\SFS$ has three terms:
        \begin{equation*}
            \vol \bigl( (\partial \Omega_\SFS)_\delta^+ \bigr) = \area (\partial \Omega_\SFS) \, \delta + 3 \, W_2 (\Omega_\SFS) \, \delta^2 + \omega_3 \, \delta^3 \ .
        \end{equation*}
        As shown in \cref{exa:MinkowskiSteiner3D}, we have $W_2 (\Omega_\SFS) \le \omega_3 \, \diam(\Omega_\SFS)$ because $0 \in \Omega_\SFS$ by the assumption $B^d \subset \Omega_\SFS$. Since also $\omega_3 \le W_2 (\Omega_\SFS)$ as argued in \cref{rem:refinedVolumeEstimate}, it follows that the two estimates for $\Gamma_{2 \delta} (\Omega_\SFS)$ from \cref{eq:refinedVolumeEstimate} become
        \begin{equation*}
            \delta \cdot \area (\partial \Omega_\SFS) + 3 \omega_3 \, \delta^2 + \omega_3 \, \delta^3 \le \vol \bigl( \Gamma_{2 \delta} (\Omega_\SFS) \bigr) \le 2 \delta \cdot \area (\partial \Omega_\SFS) + 3 \omega_3 \, \diam(\Omega_\SFS) \, \delta^2 + \omega_3 \, \delta^3 \ .
        \end{equation*}
        Using this inequality instead of \cref{eq:ratioEstimate} in the proof of \cref{pro:STVR}, the upper bound on the ratio of energies in \cref{eq:energyEstimate} turns into
        \begin{equation*}
             \frac{\EV_\rho [H_\sint]}{\EV_\rho [H_\SFS \otimes \id_\SFR]} \le C_2 (n_\SFS) \left[ 2 \kappa_\SFS \cdot \delta + \frac{3 \omega_3 \, \diam(\Omega_\SFS) \, \delta^2}{V_\SFS} + \frac{\omega_3 \, \delta^3}{V_\SFS} \right]^2 \ ,
        \end{equation*}
        while the lower bound remains structurally the same, involving only the ratio $\omega_3 / V_\SFS \ll 1$ times powers of $\delta$. In a concrete application where $\Omega_\SFS$ and $\delta$ are specified explicitly, the above estimate provides a computable error bound for the approximation in \cref{eq:STVR}. Generally speaking, however, we are still able to make the following qualitative observation: aside from the second-order correction $(\kappa_\SFS \cdot \delta)^2$ that is not modified compared to \cref{eq:energyEstimate}, the next-order term is now given by
        \begin{equation*}
            12 \omega_3 \cdot (\kappa_\SFS \cdot \delta) \cdot \frac{\diam(\Omega_\SFS) \, \delta^2}{\vol(\Omega_\SFS)} \ .
        \end{equation*}
        According to our \emph{a priori} assumption $\kappa_\SFS \cdot \delta \ll 1$, whether or not the above expression will be small is controlled by $\diam(\Omega_\SFS) \, \delta^2 / \vol(\Omega_\SFS)$. A sufficient condition ensuring that this ratio is negligible is to assume that the region $\Omega_\SFS$ is geometrically regular enough to contain a ball of radius $p^{-1} \diam(\Omega_\SFS)$, $p \in \N$, and to satisfy $\diam(\Omega_\SFS) \gg 1$ in addition to $V_\SFS \gg 1$ from \ref{enu:A1.1}. Then $\vol(\Omega_\SFS) \ge \omega_3 \, \diam(\Omega_\SFS)^3 / p^3$ and hence
        \begin{equation*}
            \frac{\diam(\Omega_\SFS) \, \delta^2}{\vol(\Omega_\SFS)} \le \frac{p^3}{\omega_3} \left( \frac{\delta}{\diam(\Omega_\SFS)} \right)^2 \ll 1 \ .
        \end{equation*}

        \item Consider again an arbitrary $d \in \N$ but suppose that $\Omega_\SFS = B^d(0, R)$ is a ball of radius $R \gg 1$. Then $\partial \Omega_\SFS = S^{d-1}(0, R)$ is the sphere of radius $R$, and hence the interaction corridor is given by (see also Ref.~\cite[p. 1189]{Osserman1978})
        \begin{equation*}
            \Gamma_{2 \delta} (\Omega_\SFS) = (\partial \Omega_\SFS)_\delta = B^d (0, R + \delta) \setminus B^d (0, R - \delta) \ .
        \end{equation*}
        Its volume can be computed explicitly:
        \begin{align*}
            \vol \bigl( \Gamma_{2 \delta} (\Omega_\SFS) \bigr) &= \omega_d \Bigl[ (R + \delta)^d - (R - \delta)^d \Bigr] = 2 \omega_d \sum_{\substack{1 \le k \le d \\ \text{$k$ odd}}} \binom{d}{k} R^{d - k} \delta^k \ ,
        \end{align*}
        where the binomial theorem was used. Since $\vol(\Omega_\SFS) = \omega_d R^d$, it follows that
        \begin{equation}\label{eq:volumeEstimateShere}
            \frac{\vol \bigl(\Gamma_{2 \delta} (\Omega_\SFS)\bigr)}{\vol(\Omega_\SFS)} = 2 \sum_{\substack{1 \le k \le d \\ \text{$k$ odd}}} \binom{d}{k} R^{- k} \delta^k \ .
        \end{equation}
        The crucial physical criterion guaranteeing a negligible surface-to-volume ratio is
        \begin{equation*}
            \frac{\delta}{R} \ll 1
        \end{equation*}
        because in this case, the estimate \eqref{eq:energyEstimate} in the proof of \cref{pro:STVR} becomes
        \begin{equation*}
            C_1 (n_\SFS) \left[ 4 d^2 \left(\frac{\delta}{R}\right)^2 + \MO \left( \frac{\delta}{R} \right)^4 \right] \le \frac{\EV_\rho [H_\sint]}{\EV_\rho [H_\SFS \otimes \id_\SFR]} \le C_2 (n_\SFS) \left[ 4 d^2 \left(\frac{\delta}{R}\right)^2 + \MO \left( \frac{\delta}{R} \right)^4 \right] \ .
        \end{equation*}
        Suppose now that $d = 3$. In this case, there are only two terms in \cref{eq:volumeEstimateShere}; hence, we can write down the summands in the square brackets appearing in the upper and lower bound, obtained by squaring \cref{eq:volumeEstimateShere}, explicitly:
        \begin{equation*}
            36 \left(\frac{\delta}{R}\right)^2 + 24 \left(\frac{\delta}{R}\right)^4 + 4 \left(\frac{\delta}{R}\right)^6 \ .
        \end{equation*}
    \end{enumerate}
\end{examples}

\section{Varying particle number and Fock space}\label{sec:Fock}

In the next step, we want to add to our open-system model, determined so far in terms of Assumptions \ref{asm:systemDivision} to \ref{enu:A4}, the additional physical feature of a varying particle number. The statistical mechanics literature (see, e.g., Refs.~\cite{BratteliRobinson1981, Haag1996, Ruelle1999}) usually takes for granted that the Hilbert space of systems with this property is given by Fock space \cite{Fock1932, Cook1953}, while a particle number operator is defined \emph{a posteriori} only. Here we shall follow a conceptually different approach: first a new assumption will be introduced, \emph{postulating the existence} of a particle number observable, and second it will be shown that this assumption \emph{necessitates Fock space}. This reversal is essential for a first-principles approach, as the physical observable should be considered more fundamental than choosing a (convenient) representation.

\subsection{Particle number operator}

As indicated above, let us start by posing the following assumption, which implements the physical feature of a varying particle number in terms of an observable with specific mathematical properties.

\setcounter{assumption}{5}
\begin{assumption}[Varying particle number]\label{asm:variableParticleNumber}
    Suppose that the open system $\SFS$ consists of identical particles, let $\MH$ denote the one-particle Hilbert space of this species, and assume that the total particle number $N_\SFS$ is varying. Then we pose that
    \begin{enumerate}[labelindent=\parindent, leftmargin=*, label=\normalfont\textbf{A6.1.}, ref=A6.1]
        \item \label{enu:A6.1} $\MH^{\otimes n} \subset \MH_\SFS$ for all $n \in \N_0$ with isometric embedding. (We set $\MH^{\otimes 0} \ce \C$.)
    \end{enumerate}
    According to the standard description of quantum mechanics, the Hilbert space $\MH^{\otimes n} \ce \bigotimes^n \MH$ is the state space for an $n$-particle system (recall \cref{rem:indistinguishability} in this context), and hence the space $\MH_\SFS$ should contain all $\MH^{\otimes n}$ in order to be able to describe states with different particle numbers; the embeddings should, moreover, be isometric to guarantee that a pure state (i.e., a normalized vector) in $\MH^{\otimes n}$ is also a pure state in $\MH_\SFS$. Next, we postulate:
    \begin{enumerate}[labelindent=\parindent, leftmargin=*, label=\normalfont\textbf{A6.2.}, ref=A6.2]
        \item \label{enu:A6.2} There exists a self-adjoint operator $\MSN \in \LO(\MH_\SFS)$ such that (i) $\sigma(\MSN) = \sigma_\mathrm{p}(\MSN) = \N_0$ and (ii) $\Eig(\MSN, n) = \MH^{\otimes n}$ for all $n \in \N_0$.
    \end{enumerate}
    Since the particle number is not merely a parameter anymore but a physical observable, there should be a self-adjoint linear operator representing this observable, the spectral values of which are given by the possible particle numbers, and the corresponding eigenspaces are the $n$-particle subspaces $\MH^{\otimes n}$.
\end{assumption}

\begin{remark}\label{rem:variableParticleNumber}
    In ordinary, non-relativistic quantum mechanics, the particle number operator is typically defined in terms of creation and annihilation operators, and the latter are used to generate and remove \emph{excitations} of physical particles whose number within the open system is fixed, but not to exchange the particles themselves with the reservoir. A truly varying particle number usually appears in quantum field theory, where the creation and annihilation of elementary particles due to vacuum fluctuations are described. In light of these two situations mostly treated in the literature, it must be emphasized that we are treating existing physical particles, which constitute actual matter and not merely excitations, that can enter and exit the open system, but are not created or destroyed as in quantum field theory. For a more detailed discussion of the difference between the exchange of excitations and the exchange of physical particles in the context of the Lindblad equation, we refer to Ref.~\cite{Reible2025pre}.
\end{remark}

Based on the physically motivated \cref{asm:variableParticleNumber}, the following result is a straightforward consequence of the spectral theorem, but nevertheless a conceptually very important step towards the effective Hamiltonian.

\begin{proposition}[Necessity of Fock space]\label{pro:FockSpace}
    Assume that the open system $\SFS$ satisfies \cref{asm:variableParticleNumber}. Then its Hilbert space $\MH_\SFS$ is given by Fock space:
    \begin{equation}\label{eq:FockSpace}
        \MH_\SFS = \bigoplus_{n=0}^{\infty} \MH^{\otimes n} \ .
    \end{equation}
\end{proposition}

\begin{proof}
    Let $E_\MSN$ denote the unique spectral measure of the self-adjoint operator $\MSN \in \LO(\MH_\SFS)$ from Assumption \ref{enu:A6.2}, which exists according to the spectral theorem, is defined on the Borel $\sigma$-algebra of $\R$, and supported on the spectrum $\sigma(\MSN)$. Since $\sigma(\MSN) = \N_0 = \bigcup_{n \in \N_0} \, \set{n}$ is a countable disjoint union and $E_\MSN$ is $\sigma$-additive, it holds that
    \begin{equation*}
        \id_{\MH_\SFS} = E_\MSN \bigl( \sigma(\MSN) \bigr) = E_\MSN \biggl( \bigcup\nolimits_{n \in \N_0} \set{n} \biggr) = \sum_{n=0}^{\infty} E_\MSN \bigl( \set{n} \bigr) \ ,
    \end{equation*}
    where the series on the right converges in the strong operator topology; this is to say that for all $\xi \in \MH_\SFS$, the following series converges in the Hilbert space norm:
    \begin{equation*}
        \xi = \sum_{n=0}^{\infty} \xi_n \quad \text{with} \quad \xi_n \ce E_\MSN \bigl( \set{n} \bigr) \xi \ .
    \end{equation*}
    As the operator $E_\MSN \bigl( \set{n} \bigr)$ is the orthogonal projection onto $\Eig(\MSN, n)$, and as different eigenspaces are mutually orthogonal, it follows that also the vectors $(\xi_n)_{n \in \N_0}$ are pairwise orthogonal. Therefore, the Pythagorean theorem and continuity of the norm imply that
    \begin{equation*}
        \norm{\xi}^2 = \sum_{n=0}^{\infty} \norm{\xi_n}^2 \ .
    \end{equation*}
    Note that the right-hand side is equal to the canonical norm of the Hilbert space direct sum $\bigoplus_{n \in \N_0} \Eig(\MSN, n)$ of the subspaces $\Eig(\MSN, n) = \MH^{\otimes n} \subset \MH_\SFS$, cf. \cref{eq:directSumNorm} in \cref{app:directSum}. Thus, using the previous two identities for $\xi \in \MH_\SFS$, it follows that the mapping $J : \MH_\SFS \to \bigoplus_{n \in \N_0} \MH^{\otimes n}$ given by $J \xi \ce \bigl( E_\MSN (\set{n}) \xi)_{n \in \N_0}$ is an isometric isomorphism, and hence the open system's Hilbert space $\MH_\SFS$ is isometrically isomorphic to the Fock space $\bigoplus_{n \in \N_0} \MH^{\otimes n}$.
\end{proof}

In \cref{app:directSum}, we provide another proof of this proposition that uses the universal property of the Hilbert space direct sum, and in the following, we comment on how our approach to Fock space complements the existing literature.

\begin{remarks}\label{rem:FockSpaceAQFT}
    \leavevmode
    \begin{enumerate}[env]
        \item \label{enu:remFockSpaceAQFT} The assertion of \cref{pro:FockSpace} may be summarized as follows: if there exists a total particle number operator on $\MH_\SFS$, then $\MH_\SFS$ has to be Fock space. In this context, it must be noted that the connection between the existence of a number operator and Fock space has been studied intensively in the literature on representations of the canonical (anti-)commutation relations, e.g., \cite{Chaiken1967, Chaiken1968, Courbage1971, DellAntonio1967, DellAntonio1966, Friedrichs1952}; see, in particular, \cite[pp. 23--25]{Chaiken1967} for a summary as well as the textbook accounts \cite[pp. 30 ff.]{BratteliRobinson1981} and \cite[pp. 21 f.]{Strocchi2013}. In essence, these studies showed that an operator of the form $\MSN = \sum_{k=1}^{\infty} a_k^\ast a_k$, where $a_k^\ast$ and $a_k$ are the creation and annihilation operators, exists only in a direct sum of Fock representations of the canonical (anti-)commutation relations. Compared to this, the novelty of \cref{pro:FockSpace}, aside from the one discussed in \cref{rem:variableParticleNumber}, is that it did neither require the abstract framework of algebraic quantum field theory nor a mathematically involved definition of the number operator in terms of creation and annihilation operators; instead, the characterization of $\MSN$ through \cref{asm:variableParticleNumber} is mathematically simple but also physically well motivated.

        \item In recent years, there has been increased interest in the physical feature of a varying particle number in the framework of classical mechanics; see, for example, Refs.~\cite{delRazo2025, delRazo2022, DelleSite2020, Klein2022} and references therein. Hence, it is worthwhile to discuss \cref{pro:FockSpace} in this context. Since there is no canonical notion of a classical Fock space, one has to model the configuration space of a classical system with varying particle number in terms of a hierarchy of hyperplanes of configuration spaces with fixed particle number. These hyperplanes are connected by dynamical transition probabilities (free energy), which are determined based on statistical mechanics considerations. Thus, \cref{pro:FockSpace} shows that in quantum mechanics the mathematical framework for systems with varying particle number is more compact (which is also reflected by \cref{cor:vonNeumannEquation} below). However, it should be noted that Ref.~\cite{delRazo2022} put forward the proposal to use a Fock space formalism in a classical framework, suitably adapted to the physical situation, to model particle-based reaction-diffusion processes.
    \end{enumerate}
\end{remarks}

If the physics dictates an upper bound on the particle number, then $\MH_\SFS$ will only be a finite direct sum. This is proved in the same way as \cref{pro:FockSpace} above by replacing the spectrum $\N_0$ with the set $\set{0, 1, \dotsc, M}$.

\begin{corollary}
    If the total particle number $N_\SFS$ of $\SFS$ is varying and bounded above by $M \in \N$, then the Hilbert space of $\SFS$ is given by the finite direct sum
    \begin{equation*}
        \MH_\SFS = \bigoplus_{n=0}^{M} \MH^{\otimes n} \ .
    \end{equation*}
\end{corollary}

\subsection{Direct sum operators}

Before moving on, we recall how the Hamiltonian $H_\SFS$ of the open system changes due to the transition from the $N_\SFS$-fold tensor product space \eqref{eq:HilbertSpaceS} to Fock space \eqref{eq:FockSpace}, denoted in the following as
\begin{equation*}
    \MF(\MH) \ce \bigoplus_{n=0}^{\infty} \MH^{\otimes n} \ .
\end{equation*}

Recall the general definition of the Hilbert space direct sum, provided in \cref{eq:directSum} in \cref{app:directSum}, and let $H_{\SFS, n} \in \LO(\MH^{\otimes n})$ denote the Hamiltonian of an $n$-particle realization of the open system $\SFS$, that is, $H_{\SFS, n}$ is given by \cref{eq:HamiltonianS,eq:interactionHamiltonians}, with $N_\SFS$ replaced by $n$. The appropriate Hamiltonian for $\SFS$ is the operator $\MSH \in \LO \bigl(\MF(\MH)\bigr)$ defined on \cite[p. 195]{Arai2024}
\begin{equation*}
    \dom(\MSH) \ce \Set{\bigl(\psi^{(n)}\bigr)_{n \in \N_0} \in \MF(\MH) \ : \ \psi^{(n)} \in \dom(H_{\SFS, n}) \ \text{and} \ \sum_{n \in \N_0} \norm[\big]{H_{\SFS, n} \psi^{(n)}}_{\MH^{\otimes n}}^2 < + \infty} \ .
\end{equation*}
and acting on an element $\Psi = (\psi^{(n)})_{n \in \N_0} \in \dom(\MSH)$ by $\MSH \Psi \ce (H_{\SFS, n} \psi^{(n)})_{n \in \N_0}$; due to this definition, one denotes the operator $\MSH$ also by
\begin{equation*}
    \MSH = \bigoplus_{n=0}^{\infty} H_{\SFS, n} \ .
\end{equation*}
One can show that if the operators $H_{\SFS, n}$ are all self-adjoint (which we assumed), then $\MSH$ is also self-adjoint \cite[Thm. 4.2 (vi)]{Arai2024}. Similarly to the definition of $\MSH$, the number operator $\MSN \in \LO(\MH_\SFS)$ introduced in Assumption \ref{enu:A6.2} may be written as (cf. Ref.~\cite[p. 221]{Arai2024})
\begin{equation*}
    \MSN = \bigoplus_{n=0}^{\infty} \bigl( n \id_{\MH^{\otimes n}} \bigr) \ .
\end{equation*}

In addition to the change of the Hamiltonian $H_\SFS$ of the open system to the direct sum operator $\MSH$, let us also adapt the set of admissible density operators---first introduced in \cref{eq:admissibleDM}---to the situation of varying particle number. On one hand, within the first-order approximation \eqref{eq:STVR}, that is, up to terms of higher order in the effective interaction distance $\delta$, the operator $H_\sint$ mediating the system--reservoir interaction does not play a role anymore (\cref{pro:STVR}). Since we will employ this approximation in the sequel (and in the next subsection, we will argue that it is indeed still valid for systems with varying particle number), we can drop the compatibility requirement of our admissible density operators with respect to the Hamiltonian $H_\sint$. On the other hand, the physically relevant density operators should now be compatible with the particle number operator $\MSN$ as well due to \cref{asm:variableParticleNumber}. Therefore, we define
\begin{equation*}
    \DM_\mathrm{adm}^{(2)} (\MH_\SFT) \ce \DM(\MH_\SFT, \MSH \otimes \id_\SFR) \cap \DM(\MH_\SFT, \MSN \otimes \id_\SFR) \cap \DM(\MH_\SFT, \id_\SFS \otimes H_\SFR) \ .
\end{equation*}

\subsection{Validity of the surface-to-volume ratio approximation}\label{subsec:STVRvariableN}

When deriving the surface-to-volume ratio approximation of the total Hamiltonian $H_\SFT$ in \cref{sec:STVR}, we have treated the particle number $N_\SFS$ of the open system as a constant. Since we are now considering systems satisfying \cref{asm:variableParticleNumber}, the quantity $N_\SFS$ will fluctuate, and hence it is not clear whether \cref{pro:STVR} is still valid for these systems; in the following, we will argue that, under a natural requirement, it is indeed.

The crucial observation is that we are still working within the framework of equilibrium statistical mechanics (\cref{asm:systemDivision}). \emph{Firstly}, this entails the assumption that the number of particles inside $\SFS$ is large \emph{on average}, i.e., that the variable $N_\SFS$ now fluctuates around its large mean value $\oln{N_\SFS} \gg 1$, in accordance with \cref{asm:variableParticleNumber}. \emph{Secondly}, this implies that we are not treating the operator $\MSN$ as a fully resolved dynamical variable of the system whose time evolution can be determined exactly, but rather as a thermodynamic property that undergoes \emph{statistical fluctuations} and, in this way, identifies different \emph{statistical realizations} (characterized in terms of different particle numbers) of the open system $\SFS$.

Suppose that $r \gg 1$ consecutive, independent measurements of the particle number $N_\SFS$ have been performed on $\SFS$, and that $N_\SFS^{(1)}, \dotsc, N_\SFS^{(r)}$ are the detected outcomes, i.e., the different statistical realizations of the open system. Let $\braket{N_\SFS} \ce \frac{1}{r} \sum_{j=1}^{r} N_\SFS^{(j)}$ denote the sample mean and $\sigma^2 = \braket{N_\SFS^2} - \braket{N_\SFS}^2$ the sample variance. \emph{In principle}, the latter could be arbitrarily large, for example, if one starts with a large initial particle number $N_\SFS^{(1)} \gg 1$ satisfying Assumption \ref{enu:A1.1} but, in the course of the statistical evolution, also obtains realizations $j \in \set{1, \dotsc, r}$ for which $N_\SFS^{(j)} / N_\SFS^{(1)} \ll 1$. \emph{However}, in the setting of equilibrium statistical mechanics, such occurrences are prohibited since the relative mean square fluctuation of thermodynamic observables, like the open system's particle number, must be very small, i.e.,
\begin{equation}\label{eq:smallMSF}
    \frac{\sigma}{\braket{N_\SFS}} = \sqrt{\frac{\braket{N_\SFS^2} - \braket{N_\SFS}^2}{\braket{N_\SFS}^2}} \ll 1 \ .
\end{equation}
In the present form, \cref{eq:smallMSF} is a basic physical requirement, which could be translated into a more precise mathematical condition on the admissible density operators $\rho$ of the total system, that reconciles our previous Assumptions \ref{asm:systemDivision} and \ref{asm:variableParticleNumber} with each other. As indicated above, \cref{eq:smallMSF} should really be understood as a general consequence of the statistical-mechanical description that we are using as the underlying theoretical framework for our derivations, ensuring that the ensemble average and most probable value of $\MSN$ are approximately equal \cite[p. 129]{Huang1987}, \cite[pp. 18, 86]{Nolting2018}. In fact, as K.~Huang writes \cite[p. 130]{Huang1987}:
\begin{itemize}
    \item[] \enquote{If this condition is not satisfied, there is no unique way to determine how the observed value of [the measurable quantity $\MSN$]\footnote{Modification by the authors.} may be calculated. When it is not, we should question the validity of statistical mechanics.}
\end{itemize}

Note that for the concrete example of the grand canonical density operator, one can calculate the relative mean square fluctuation of the particle number explicitly, which gives the following qualitative behavior \cite[Eq. (1.200)]{Nolting2018}:
\begin{equation*}
    \frac{\sigma}{\braket{N_\SFS}} \sim \sqrt{\frac{\kappa_T}{\braket{N_\SFS}}} \ .
\end{equation*}
Here $\kappa_T$ is the isothermal compressibility, which is finite outside the neighborhood of a first-order phase transition, and thus for large $\braket{N_\SFS}$ the fluctuations around the mean are negligibly small; this illustrates that small fluctuations are indeed a natural condition.

Let us return to our argument why \cref{pro:STVR} is valid also for varying $N_\SFS$. Looking back at the proof of the latter, one recognizes that the variable $N_\SFS$ enters explicitly in conjunction with Assumption \ref{enu:A1.1}, which is not invalidated by \cref{asm:variableParticleNumber} as explained above, and implicitly through the dependence of the constants $C_1(n_\SFS)$ and $C_2(n_\SFS)$ on the number density of the open system, which fluctuates now as well. By virtue of \cref{eq:smallMSF}, however, we know that, statistically speaking, $n_\SFS$ will not deviate significantly from its mean value $\oln{n_\SFS}$, and hence we can perform the following additional approximations:
\begin{equation*}
    C_1(n_\SFS) \approx C_1(\oln{n_\SFS}) \quad \text{and} \quad C_2(n_\SFS) \approx C_2(\oln{n_\SFS}) \ .
\end{equation*}
Since the exact value of the quantities $C_1$ and $C_2$ is not relevant for the validity of the surface-to-volume ratio approximation \eqref{eq:STVR} even in the case of a constant particle number, this newly introduced approximation is not a drastic one. Therefore, \cref{pro:STVR} also applies to systems that are characterized by \cref{asm:variableParticleNumber}.

\section{The effective Hamiltonian}\label{sec:effectiveHamiltonian}

We have established a rigorous version of the surface-to-volume ratio approximation (\cref{pro:STVR}) as well as the necessity of Fock space for systems with varying physical particle number (\cref{pro:FockSpace}). Based on these results, we can now derive the effective Hamiltonian for the open system $\SFS$ and argue that, given the necessary assumptions and approximations, it is unique up to a constant.

\subsection{Main result}

In the sequel, we will build on ideas put forward in Ref.~\cite{DelleSite2024jpa}, where a similar effective Hamiltonian was derived based on semi-empirical assertions, but we will follow a different, more rigorous approach; see \cref{subsec:discussionEffectiveHamiltonian} below for a discussion of the differences of the two approaches.

The following theorem constitutes the main result of our paper. It relies on Assumption \ref{asm:variableParticleNumber} and on the surface-to-volume ratio approximation \eqref{eq:STVR} from \cref{pro:STVR}. We underline, therefore, that the assertion must be understood within a first-order approximation in the effective interaction distance $\delta$, and that if one were to consider higher-order terms in \cref{eq:STVR}, then this assertion would change.

\begin{theorem}[Effective Hamiltonian]\label{thm:effectiveHamiltonian}
    Suppose that the surface-to-volume ratio approximation has been applied to the total system $\SFT$, and that the open system $\SFS$ satisfies Assumption \ref{asm:variableParticleNumber}. Set $\varepsilon \ce N_\SFS / N$. Then there exists a constant $\mu \in \R$ such that
    \begin{equation}\label{eq:effectiveHamiltonian}
        H_\SFT = (\MSH - \mu \MSN) \otimes \id_{\SFR} + \MO (\varepsilon^2) \ .
    \end{equation}
    The operator $\MSH_\seff \ce \MSH - \mu \MSN$ is called the \emph{effective Hamiltonian} for the open system $\SFS$, and it holds that within a first-order approximation with respect to $\varepsilon$, this operator is the unique quantum-mechanical Hamiltonian of $\SFS$, up to additive constants.
\end{theorem}

Before we come to the proof of the theorem, we shall give two remarks about the interpretation of the result.

\begin{remarks}\label{rem:effHamiltonian}
    \leavevmode
    \begin{enumerate}[env]
        \item Let us first emphasize in what sense the uniqueness statement should be understood. If we approximate the exact total Hamiltonian \eqref{eq:totalHamiltonian} up to an error of second order in $\delta$ using the surface-to-volume ratio approximation from \cref{pro:STVR}, and if we subsequently approximate the resulting operator up to an error of second order in $\varepsilon$ as well, then \emph{within these two first-order approximations}, $\MSH_\seff$ is the unique Hamiltonian of the open system, modulo additive constants. Thus, since the operator $\MSH_\seff$ is only correct up to terms of higher order, it is called an \emph{effective} Hamiltonian, meaning an \emph{approximate} one.

        \item \label{enu:epsilon} Next, let us take a look at the parameter $\varepsilon$. According to \cref{asm:variableParticleNumber} the particle number $N_\SFS$ fluctuates and hence also $\varepsilon = N_\SFS / N$ will be a variable. This circumstance can be understood as follows: \cref{eq:effectiveHamiltonian} is an expansion of the function $N_\SFS \mapsto H_\SFT(N_\SFS)$ in terms of the variable $N_\SFS$; indeed, as we will see in the proof below, the total Hamiltonian $H_\SFT$ depends, through the \emph{reservoir's Hamiltonian} $H_\SFR$, on the particle number of the open system. Consequently, \cref{eq:effectiveHamiltonian} shows that, to first order, this dependency is captured by the number operator $\MSN$. In addition, it must be noted that even though $\varepsilon$ fluctuates, it will always---that is, for all statistical realizations of the open system characterized in terms of different particle numbers---be very small because by \cref{eq:smallMSF} the relative mean square fluctuation of $\MSN$ is negligible and hence $N_\SFS$ will be close to its very large mean value.
    \end{enumerate}
\end{remarks}

\begin{proof}
    Let $\rho \in \DM_\mathrm{adm}^{(2)} (\MH_\SFT)$ be an admissible density operator for the total system and define $E_\SFT \ce \EV_\rho [H_\SFT]$. According to the surface-to-volume ratio approximation \eqref{eq:STVR}, up to first order in the effective interaction distance $\delta$, we have
    \begin{align}\label{eq:expectedTotalEnergyFock}
        \begin{split}
            E_\SFT &\approx \EV_\rho [\MSH \otimes \id_\SFR + \id_\SFS \otimes H_\SFR] \\
            &= \EV_\rho [\MSH \otimes \id_\SFR] + \EV_\rho [\id_\SFS \otimes H_\SFR] = \EV_{\rho_\SFS} [\MSH] + \EV_{\rho_\SFR} [H_\SFR] \ ,
        \end{split}
    \end{align}
    where we have defined $\rho_\SFR \ce \tr_\SFS (\rho)$ to be the partial trace of the state $\rho$ with respect to the open system and \cref{lem:partialTrace} was used in the last step. Recall from \cref{subsec:STVRvariableN} that \cref{pro:STVR} remains valid for varying $N_\SFS$, justifying the approximation \eqref{eq:expectedTotalEnergyFock}.
    
    Let $E_\SFR \ce \EV_{\rho_\SFR} [H_\SFR]$ be the expected energy of the reservoir. The central physical observation for our proof, which was the key modeling assumption in Ref.~\cite[p. 5]{DelleSite2024jpa}, is the following: $E_\SFR$ is a function of the number of particles $N_\SFR$ of the reservoir, given by $N_\SFR = N - N_\SFS$ (cf. \cref{asm:systemDivision}); while this is physically plausible, it can be seen mathematically by writing down an expression like \eqref{eq:HamiltonianS} for the Hamiltonian $H_\SFR$ and then computing the expectation value as in the proof of \cref{pro:STVR}. Thus, letting $\varepsilon = N_\SFS / N$, it follows that
    \begin{equation*}
        E_\SFR (N_\SFR) = E_\SFR (N - N_\SFS) = E_\SFR \bigl( (1 - \varepsilon) N \bigr) \ .
    \end{equation*}
    Since $\varepsilon \ll 1$ is small for all statistical realizations of the open system according to Assumption \ref{enu:A1.1} and \cref{rem:effHamiltonian} \ref{enu:epsilon} from above, we can expand the function $\varepsilon \mapsto E_\SFR \bigl( (1 - \varepsilon) N \bigr)$ into a Taylor series around $\varepsilon = 0$:
    \begin{equation}\label{eq:expansionER}
        E_\SFR \bigl( (1 - \varepsilon) N \bigr) = E_\SFR (N) - \varepsilon N \left. \od{E_\SFR (N_\SFR)}{N_\SFR} \right|_{N_\SFR=N}  + \MO (\varepsilon^2) \ .
    \end{equation}
    Of course, the validity of this expansion requires $N_\SFR \mapsto E_\SFR (N_\SFR)$ to be at least a $C^1$-function; since we did not specify a concrete reservoir model, we have to pose this as an additional assumption for the validity of the theorem. It must be noted, however, that differentiability of $E_\SFR$ is a physically very reasonable hypothesis since if it were not given, there would be no thermodynamics for $\SFR$, and hence reservoir models that do not satisfy this assumption are physically uninteresting.
    
    Looking now at the expansion \eqref{eq:expansionER}, the zeroth-order term $E_\SFR (N)$ is a constant that can be considered a shift of the energy scale; without loss of generality, we may neglect it by shifting $H_\SFR$ similarly as $H_\SFS$ in \eqref{eq:energyShift}. The coefficient of the first-order term, on the other hand, is an important physical quantity. By definition, it corresponds to the \emph{chemical potential} $\mu_\SFR$ of the reservoir at constant entropy and volume \cite[Eq. (2)]{Ali2020}, \cite[Eq. (7.31)]{Huang1987}, \cite[Eq. (3.1.3)]{Schwabl2006}:
    \begin{equation}\label{eq:chemicalPotential}
        \mu_\SFR = \left. \od{E_\SFR (N_\SFR)}{N_\SFR} \right|_{N_\SFR=N} \ .
    \end{equation}
    Now, crucially, since we have assumed that the open system and the reservoir have reached thermodynamic equilibrium (Assumption \ref{enu:A1.2}), it follows that the chemical potential $\mu \equiv \mu_\SFS$ of the open system satisfies $\mu = \mu_\SFR$ \cite[Eq.~(2.7.4)]{Schwabl2006}. Moreover, as the open system is macroscopic (Assumption \ref{enu:A1.1}), the particle number $N_\SFS$ is given by the expectation value $\EV_{\rho_\SFS} [\MSN] = \tr_{\MH_\SFS} (\MSN \! \rho_\SFS)$ of the particle number operator \cite[p. 83]{Nolting2018}, \cite[p. 65]{Schwabl2006}:
    \begin{equation*}
        N_\SFS = \tr_{\MH_\SFS} (\MSN \! \rho_\SFS) \ .
    \end{equation*}
    Note that by \cref{lem:partialTrace} and the assumption $\rho \in \DM (\MH_\SFT, \MSN \otimes \id_\SFR)$, this expectation value is indeed well-defined. With all of these observations, the expansion \eqref{eq:expansionER} reduces to
    \begin{equation}\label{eq:expectedEnergyReservoir}
        \EV_{\rho_\SFR} [H_\SFR] = - \tr_{\MH_\SFS} ( \mu \MSN \! \rho_\SFS) + \MO (\varepsilon^2) \ .
    \end{equation}

    Let us return to the expected energy $E_\SFT = \EV_\rho [H_\SFT]$ of the total system in the state $\rho \in \DM_\mathrm{adm}^{(2)} (\MH_\SFT)$. Since the Hamiltonian $H_\SFT$ corresponds to the observable of the total energy, it generates the total system's time evolution \cite[Axiom A6, pp. 794 f.]{Moretti2017}, \cite[Axiom 4, p. 65]{Teschl2014} and is, therefore, unique up to additive constants \cite[Thm. 12.45]{Moretti2017}. That is, if $Q \in \LO (\MH_\SFT)$ were another self-adjoint operator such that $E_\SFT = \EV_\rho [Q]$ for all admissible states $\rho$, i.e., such that $Q$ corresponds to the observable of the total energy of $\SFT$, then it would follow that $Q = H_\SFT + C \id_\SFT$ for some $C \in \R$. Now, inserting the previous result \eqref{eq:expectedEnergyReservoir} into the formula \eqref{eq:expectedTotalEnergyFock} for the approximated $E_\SFT$, it follows that
    \begin{align*}
        E_\SFT &\approx \EV_{\rho_\SFS} [\MSH] + \EV_{\rho_\SFR} [H_\SFR] \\
        &= \tr_{\MH_\SFS} (\MSH \! \rho_\SFS) - \tr_{\MH_\SFS} ( \mu \MSN \! \rho_\SFS) + \MO (\varepsilon^2) \\
        &= \tr_{\MH_\SFS} \bigl( (\MSH - \mu \MSN) \rho_\SFS \bigr) + \MO (\varepsilon^2) \ .
    \end{align*}
    It holds that $\rho \in \DM (\MH_\SFT, \MSH - \mu \MSN)$ by assumption on $\rho$ and \cref{lem:intersectionCompatibleDM}. Hence, we can apply \cref{lem:partialTrace} to rewrite the expectation value in the last line as
    \begin{equation*}
        \tr_{\MH_\SFS} \bigl( (\MSH - \mu \MSN) \rho_\SFS \bigr) = \tr_{\MH_\SFS \otimes \MH_\SFR} \Bigl[ \bigl( (\MSH - \mu \MSN) \otimes \id_\SFR \bigr) \rho \Bigr] \ .
    \end{equation*}
    Thus, it follows that the expected energy of the total system can be expressed as
    \begin{equation*}
        E_\SFT \approx \EV_\rho \bigl[ (\MSH - \mu \MSN) \otimes \id_\SFR \bigr] + \MO (\varepsilon^2) \ .
    \end{equation*}
    
    According to the argument for the uniqueness of the Hamiltonian provided above, we may conclude that under the relevant assumptions and the implementation of the surface-to-volume ratio approximation, the operator $(\MSH - \mu \MSN) \otimes \id_\SFR$ must be the unique first-order Hamiltonian (up to additive constants) of the total system $\SFT$.
\end{proof}

\subsection{Discussion of the result}\label{subsec:discussionEffectiveHamiltonian}

In the following, we will discuss some aspects of the above proof in detail and also how our result complements the existing literature.

\emph{Regarding the proof of \cref{thm:effectiveHamiltonian}}, we emphasize \emph{first} that its crucial element consists in treating the reservoir's energy $E_\SFR = \tr_{\MH_\SFR} (H_\SFR \rho_\SFR)$ as a function of $N - N_\SFS$, and the system's particle number $N_\SFS$ as a small perturbation of $N$. This permits a series expansion of $E_\SFR$, which corresponds to a coarse graining of the microscopic model of the reservoir since only a first-order response of $\SFR$ to changes in the particle number of $\SFS$ is kept. Therefore, $H_\SFR$ does not appear in \cref{eq:effectiveHamiltonian} anymore, and the influence of the reservoir on the open system is reduced to the chemical potential. \emph{Second}, let us comment on some of the employed assumptions. Most importantly, the derivation of \eqref{eq:effectiveHamiltonian} relied on (i) the surface-to-volume ratio approximation \eqref{eq:STVR} to rewrite $\EV_\rho [H_\SFT]$ in terms of the operator $H_\SFS \otimes \id_\SFR + \id_\SFS \otimes H_\SFR$; (ii) Assumption \ref{enu:A1.1} and \cref{eq:smallMSF} entailing that $\SFS$ is macroscopically large (yet significantly smaller than $\SFR$) and that the relative fluctuation of $N_\SFS$ is small to justify $\varepsilon \ll 1$ for all statistical realizations of $\SFS$, and to connect $N_\SFS$ with the operator $\MSN$; (iii) the thermodynamic equilibrium Assumption \ref{enu:A1.2}, which implied that $\mu = \mu_\SFR$. Dropping any of these, \cref{eq:effectiveHamiltonian} loses its validity. In fact, without the approximation (i), the chemical potential of $\SFS$ cannot be given by \cref{eq:chemicalPotential} because it is no longer a constant: if the open system is too small (in the sense that surface effects play a significant role), then the exchange of particles with the reservoir drastically changes the physics of $\SFS$ (in particular, $\mu_\SFS$ changes), and thus higher-order terms in the expansion \eqref{eq:expansionER} will become relevant; for example, the second-order term involves the change of the chemical potential with respect to the number of particles:
\begin{equation*}
    \left. \od[2]{E_\SFR}{N_\SFR} \right|_{N_\SFR=N} = \left. \od{\mu_\SFR}{N_\SFR} \right|_{N_\SFR=N} \ .
\end{equation*}
More generally, one can show \cite{Mishin2015} that if the interaction $H_\sint$ becomes important for the physics of the open system, then one cannot avoid higher-order terms in the expansion \eqref{eq:expansionER}; the same is true in case (iii) is dropped. Finally, we mention that (ii) also entails that a change of the number of particles in the open system does not affect the physics of the reservoir; in other words, the reservoir behaves as if it had a fixed particle number (see Ref.~\cite{DelleSite2024jpa} for a thorough discussion of this point).

\emph{Regarding the relation of \cref{thm:effectiveHamiltonian}} to the existing literature, let us \emph{first} underline the difference to the standard derivation of the grand canonical density operator in statistical mechanics. In the latter, the Hamiltonian $\MSH - \mu \MSN$ appears in the exponent of the density operator after the reservoir has been traced out and an expansion in the variable $N_\SFS$ similar to \eqref{eq:expansionER} has been performed. It is then \emph{defined} to be the Hamiltonian of the open system based on an analogy to the canonical density operator. Our result, in contrast, shows that one does not need to make this detour through the second-principles probability distribution to define a first-principles object: The operator $\MSH - \mu \MSN$ is derived without making use of the grand canonical ensemble, and it is shown to be the unique Hamiltonian in the specific physical regime staked out by the assumptions and approximations, describing the open system. Below, we will see that the grand canonical density operator can also be recovered from the effective Hamiltonian. \emph{Second}, let us compare the present approach to derive $\MSH_\seff$ with the previous study \cite{DelleSite2024jpa}. In the latter, a Hamiltonian of the form $H_{\SFS, n} - n \id$ was extracted from the hierarchy $\ii \hbar \dot{\rho}_n = [H_{\SFS, n} - n \id, \rho_n]$ of von Neumann equations for $n$-particle density operators $\rho_n$, $n \in \N$. To obtain this hierarchy, the framework of bounded operators was used, the surface-to-volume ratio approximation in its heuristic physical form \eqref{eq:heuristicSTVR} and the Born-Markov approximation were employed, the reservoir degrees of freedom were traced out of the von Neumann equation for the total system, and an expansion like \eqref{eq:expansionER} was performed. Thus, in comparison, our new approach does not require the von Neumann equation, the involved operators are not assumed to be bounded, the effective Hamiltonian is derived in Fock space, and the approximations \eqref{eq:STVR} and \eqref{eq:expansionER} are more controlled.

\subsection{Grand-canonical von Neumann equation}

As an application of \cref{thm:effectiveHamiltonian}, we show that the effective Hamiltonian \eqref{eq:effectiveHamiltonian} leads to a first-order von Neumann equation, which was formally derived in a similar form in Ref.~\cite{DelleSite2024jpa}.

\begin{corollary}\label{cor:vonNeumannEquation}
    In the setting of \cref{thm:effectiveHamiltonian}, suppose that the \emph{Born-Markov approximation} can be applied, i.e., for all $t \ge 0$ the time-dependent density operator $\rho(t)$ of the total system may be written as $\rho(t) = \rho_1(t) \otimes \rho_2 (0)$ for $\rho_1(t) \in \DM(\MH_\SFS)$ and $\rho_2(0) \in \DM(\MH_\SFR)$. Then the time evolution of $\rho_1(t)$ is governed by the following effective von Neumann equation:
    \begin{equation*}
        \ii \hbar \, \od{\rho_1}{t} = \bigl[ \MSH - \mu \MSN, \rho_1 \bigr] + \MO (\varepsilon^2) \ .
    \end{equation*}
    
    In particular, if $\ee^{- \beta (\MSH - \mu \MSN)} \in \BO_1(\MH_\SFS)$ is trace-class, then a stationary solution of the effective first-order von Neumann equation $\ii \hbar \, \dot{\rho}_1 = [\MSH - \mu \MSN, \rho_1]$ is given by the grand canonical density operator
    \begin{equation*}
        \rho_\mathrm{gc} = \frac{1}{Z_\mathrm{gc}} \, \ee^{- \beta (\MSH - \mu \MSN)} \quad \text{with} \quad Z_\mathrm{gc} = \tr_{\MH_\SFS} \bigl( \ee^{- \beta (\MSH - \mu \MSN)} \bigr) \ .
    \end{equation*}
\end{corollary}

\begin{proof}
    Let $\xi \in \dom(\MSH_\seff) \cap \dom(\MSH_\seff \rho_1(t))$, where $t \ge 0$, and $\eta \in \MH_\SFR$ be arbitrary. Since $H_\SFT = \MSH_\seff \otimes \id_\SFR$ (up to higher-order terms) by \cref{thm:effectiveHamiltonian}, it follows that $\xi \otimes \eta \in \dom(H_\SFT)$ and also $\rho(t) (\xi \otimes \eta) = \rho_1(t) \xi \otimes \rho_2(0) \eta \in \dom(H_\SFT)$. Thus, the following von Neumann equation holds true for the total system $\SFT$:
    \begin{equation}\label{eq:vonNeumannTotal}
        \ii \hbar \, \od{\rho(t)}{t} \, (\xi \otimes \eta) = \bigl[ H_\SFT, \rho(t) \bigr] (\xi \otimes \eta) + \MO (\varepsilon^2) \ .
    \end{equation}
    
    For every $m \in \N$ let $(H_\SFT)_m \ce (\1_{[-m,m]} \cdot \id) (H_\SFT)$, that is, $(H_\SFT)_m = \int_{[-m,m] \cap \sigma(H_\SFT)} \lambda \diff E_{H_\SFT} (\lambda)$, where $E_{H_\SFT}$ is the spectral measure of the self-adjoint operator $H_\SFT$; the latter is given by $E_{H_\SFT} = E_{\MSH_\seff} \otimes \id_\SFR$ because $H_\SFT = \MSH_\seff \otimes \id_\SFR$ \cite[Thm. 3.5 (ii)]{Arai2024}. Since the function $\1_{[-m,m]} \cdot \id$ is bounded, the previous identity implies that $(H_\SFT)_m = (\MSH_\seff)_m \otimes \id_\SFR$, see \cref{lem:functionalCalculusTP} in \cref{app:functionalCalculusTP}. Note that $(\MSH_\seff)_m$ is a bounded operator, and that $\MSH_\seff \psi = \lim_{m \to + \infty} (\MSH_\seff)_m \psi$ for all $\psi \in \dom (\MSH_\seff)$ (analogous statements hold for $H_\SFT$) \cite[Prop. 11.5]{Moretti2017}. Thus, we may rewrite the commutator on the right-hand side of \cref{eq:vonNeumannTotal} as
    \begin{align*}
        \bigl[ H_\SFT, \rho(t) \bigr] (\xi \otimes \eta) &= H_\SFT \rho(t) (\xi \otimes \eta) - \rho(t) H_\SFT (\xi \otimes \eta) \\
        &= \lim_{m \to + \infty} \Bigl( (H_\SFT)_m \rho(t) (\xi \otimes \eta) - \rho(t) (H_\SFT)_m (\xi \otimes \eta) \Bigr) \\
        &= \lim_{m \to + \infty} \Bigl( (\MSH_\seff)_m \rho_1(t) \xi \otimes \rho_2(0) \eta - \rho_1(t) (\MSH_\seff)_m \xi \otimes \rho_2(0) \eta \Bigr) \\
        &= \lim_{m \to + \infty} \Bigl( \bigl( (\MSH_\seff)_m \rho_1(t) \xi - \rho_1(t) (\MSH_\seff)_m \xi \bigr) \otimes \rho_2(0) \eta \Bigr) \\
        &= \bigl[ \MSH_\seff, \rho_1(t) \bigr] \xi \otimes \rho_2(0) \eta \ .
    \end{align*}
    Note that the algebraic manipulations of the tensor product in lines three and four were possible because the involved operators are bounded, and that in the last step continuity of the tensor product was used. With this, the von Neumann equation \eqref{eq:vonNeumannTotal} takes the form
    \begin{equation*}
        \ii \hbar \left( \od{\rho_1(t)}{t} \, \xi \otimes \rho_2(0) \eta \right) = \bigl[ \MSH_\seff, \rho_1(t) \bigr] \xi \otimes \rho_2(0) \eta + \MO (\varepsilon^2) \ .
    \end{equation*}
    Thus, since $\eta \in \MH_\SFR$ was arbitrary and $\rho_2(0) \neq 0$, it follows that for all $\xi \in \dom(\MSH_\seff) \cap \dom(\MSH_\seff \rho_1(t))$ and $t \ge 0$:
    \begin{equation*}
        \ii \hbar \, \od{\rho_1(t)}{t} \, \xi = \bigl[ \MSH_\seff, \rho_1(t) \bigr] \xi + \MO (\varepsilon^2) \ .
    \end{equation*}
    
    If we neglect the higher-order terms and assume that $\rho_1$ is stationary ($\dot{\rho}_1 = 0$), we obtain $\bigl[ (\MSH - \mu \MSN), \rho_1 \bigr] = 0$ (up to terms of order $\varepsilon^2$), which leads to $\rho_1 = \rho_\mathrm{gc}$ by following the same derivation as for the canonical ensemble, with the effective Hamiltonian $\MSH - \mu \MSN$ instead of just $H_\SFS$.
\end{proof}

\begin{remarks}
    \leavevmode
    \begin{enumerate}[env]
        \item The validity of the Born-Markov approximation relies (i) on a negligible interaction between the open system and the reservoir, which is formalized by the surface-to-volume ratio approximation (\cref{pro:STVR}), and (ii) on using a large reservoir that is not influenced by the exchange of particles with the open system, which is contained in our Assumption \ref{enu:A1.1}, suggesting that the reservoir's state can be assumed constant in time (cf. \cref{subsec:discussionEffectiveHamiltonian}). This approximation is standard in the theory of open quantum systems \cite{BreuerPetruccione2007, Manzano2020, Vacchini2024} (see, however, Refs.~\cite[Sect. 5.2.7]{RivasHuelga2012} and \cite[Sect. 2.3.3]{Rivas2010} for critical discussions), and it also fits well into the setting of \cref{thm:effectiveHamiltonian}. Moreover, the Born-Markov approximation was a key modeling assumption in Ref.~\cite{DelleSite2024jpa} for deriving the hierarchy of von Neumann equations from which the effective Hamiltonian was extracted, as discussed in \cref{subsec:discussionEffectiveHamiltonian}. Regarding the error terms, it must be noted that \cref{eq:smallMSF} ensures that the parameter $\varepsilon = N_\SFS / N$ is negligibly small for all statistical realizations of the open system. Therefore, considering only the effective first-order von Neumann equation can be considered a reasonable approximation within the physical model that we are considering.

        \item In a recent paper, we investigated the Lindblad equation for systems with varying particle number \cite{Reible2025pre}. Using the results of Ref.~\cite{DelleSite2024jpa}, we found that the system Hamiltonian should be replaced by the effective Hamiltonian $\MSH - \mu \MSN$ in order to produce a consistent Lindblad equation for such systems. Since \cref{thm:effectiveHamiltonian} justifies the conclusions of Ref.~\cite{DelleSite2024jpa} on a more rigorous level, by extension, the findings of Ref.~\cite{Reible2025pre} are also strengthened.
    \end{enumerate}
\end{remarks}

\section{Conclusion}

We have analyzed the mathematical model for open quantum systems with varying particle number and shown that the effective Hamiltonian describing such systems, under suitable assumptions and approximations, is given by the well-known expression $\MSH - \mu \MSN$. Thus, our findings constitute a proof of an empirical hypothesis used in many areas of physics that goes back at least to Bogoliubov's work in the 1950s \cite{Bogoliubov1958}. To obtain the effective Hamiltonian, we have formulated the surface-to-volume ratio approximation in a precise manner, which entails that if the open system is large enough in a certain sense, then the interaction with the reservoir can be neglected. As a key mathematical tool for establishing this approximation as well as deriving the effective Hamiltonian, we have introduced the notion of density operators compatible with unbounded observables (in the sense of computing expectation values), and we have shown that this notion leads to a satisfying framework, generalizing, for example, a partial trace formula to unbounded operators. Moreover, we have provided a new, physically motivated, and mathematically simple yet rigorous argument that for systems with varying particle number, the Hilbert space must be given by Fock space.

In summary, our findings justify some widely used tools of the trade of statistical mechanics in a solid mathematical way and, therefore, put the theory of open quantum systems with varying particle number on a more firm theoretical footing. Our results, \emph{a fortiori}, substantiate what has already been done in applications of this theory to the field of quantum technologies, where it was shown, for example, that the chemical potential (and hence the feature of a varying particle number) controls the accessible quantum states of certain simple model systems \cite{Reible2025apq}. Therefore, our work contributes to future-oriented mathematical physics that seeks to base developments at the forefront of physics and technology on mathematically solid foundations.

\appendix
\section{Compatible density operators}\label[appsec]{app:densityOperators}

\subsection{Basic Properties}

Let $T \in \LO(\MH)$ be self-adjoint. In the following, we will prove some auxiliary results about the set of $T$-compatible density operators
\begin{equation*}
    \setlength{\abovedisplayskip}{14pt}
    \setlength{\belowdisplayskip}{14pt}
    \DM(\MH, T) = \left\{ \rho \in \DM(\MH) \ : \ \parbox{24em}{\centering $\exists \, (e_n)_{n \in \N} \subset \dom(T) \ \exists \, (\alpha_n)_{n \in \N} \subset [0,1]:$ \\ (i) $(e_n)_{n \in \N}$ is an ONB for $\MH$ and $\sum_{n \in \N} \alpha_n = 1$; \\ (ii) $\rho = \sum_{n \in \N} \alpha_n P_{e_n}$; (iii) $\sum_{n \in \N} \alpha_n \abs{\braket{e_n, T e_n}} < + \infty$} \right\}
\end{equation*}
and the corresponding notion of expectation value
\begin{equation*}
    \EV_\rho[T] = \tr(T \rho) = \sum_{n=1}^{\infty} \alpha_n \braket{e_n, T e_n}
\end{equation*}
that were introduced in \cref{def:compatibleDM}.

\begin{lemma}
    For $T \in \LO(\MH)$ self-adjoint and $\rho \in \DM(\MH, T)$, the value of $\EV_\rho [T]$ does not depend on the choice of the sequences $(e_n)_{n \in \N}$ and $(\alpha_n)_{n \in \N}$ in the representation $\rho = \sum_{n \in \N} \alpha_n P_{e_n}$.
\end{lemma}

\begin{proof}
    Let $(f_m)_{m \in \N} \subset \dom(T)$ be another orthonormal basis for $\MH$ and $(\beta_m)_{m \in \N} \subset [0,1]$ be a sequence with $\sum_{m \in \N} \beta_m = 1$ such that $\rho = \sum_{m \in \N} \beta_m P_{f_m}$ and $\sum_{m \in \N} \beta_m \abs{\braket{f_m, T f_m}} < + \infty$. Using completeness of the basis $(f_m)_{m \in \N}$ and self-adjointness of $T$, it follows that
    \begin{align*}
        \sum_{n=1}^{\infty} \alpha_n \braket{e_n, T e_n} &= \sum_{n=1}^{\infty} \sum_{m=1}^{\infty} \alpha_n \braket{e_n, f_m} \braket{f_m, T e_n} = \sum_{n=1}^{\infty} \sum_{m=1}^{\infty} \alpha_n \braket{e_n, f_m} \braket{T f_m, e_n} \\
        &= \sum_{m=1}^{\infty} \braket[\Big]{T f_m, {\textstyle\sum_{n=1}^{\infty}} \alpha_n \braket{e_n, f_m} \, e_n} = \sum_{m=1}^{\infty} \braket{T f_m, \rho f_m} \\
        &= \sum_{m=1}^{\infty} \beta_m \braket{f_m, T f_m} \ .
    \end{align*}
    We note that in the second line, it was possible to interchange the two series because the sum over $m$ is finite: it holds that $\alpha_n \braket{e_n, f_m} = \braket{\rho e_n, f_m} = \braket{e_n, \rho f_m} = \beta_m \braket{e_n, f_m}$ by assumption, hence $(\alpha_n - \beta_m) \braket{e_n, f_m} = 0$ for all $n, m \in \N$. Thus, whenever $\braket{e_n, f_m} \neq 0$ we must have $\alpha_n = \beta_m$, so the sum $\sum_{m \in \N} \alpha_n \braket{e_n, f_m} \braket{T f_m, e_n}$ actually extends only over the set $I_n \ce \set{m \in \N : \beta_m = \alpha_n}$. As there are at most finitely many $m$ for which $\beta_m = \alpha_n$ (otherwise, $\sum_{m \in \N} \beta_m$ would not converge), it follows that $I_n$ is finite.
\end{proof}

\begin{lemma}\label{lem:existenceCompatibleDM}
    Let $T \in \LO(\MH)$ be self-adjoint and assume that for all $\beta > 0$, $\ee^{- \beta T} \in \BO_1(\MH)$ is a trace-class operator. Then $\DM(\MH, T)$ contains mixed states.
\end{lemma}

\begin{proof}
    First, recall the well-known fact that $\ee^{- \beta T}$ being trace-class implies that $T$ is lower semi-bounded \cite[Prop. 1.46]{Arai2024}, and that $T$ must have a purely discrete spectrum \cite[p. 597]{Moretti2017}, i.e., $\sigma(T)$ consists of isolated eigenvalues only and all eigenspaces are finite-dimensional (see also Ref.~\cite[Prop. 2.2]{Blakaj2025} for an explicit proof of these two assertions). Let $c \in \R$ be a lower bound for $T$ and define the operator $\wt{T} \ce T - c \id_\MH$. It is evident that $\wt{T}$ is positive and self-adjoint; one can show that $\wt{T}$ also has a purely discrete spectrum \cite[p. 210]{Simon1983}. Therefore, upon replacing $T$ by $\wt{T}$, we may assume, without loss of generality, that $T$ is positive.

    Since $T$ has a purely discrete spectrum, there exists a sequence $(\gamma_n)_{n \in \N} \subset [0, + \infty)$ and an orthonormal basis $(e_n)_{n \in \N}$ of $\MH$ such that $\gamma_n \to + \infty$ as $n \to + \infty$ and $T e_n = \gamma_n e_n$ for all $n \in \N$ \cite[pp. 394 f.]{Lieb1972}. Moreover, it follows from functional calculus that the assumption $\tr (\ee^{- \beta T}) < + \infty$, $\beta > 0$, is equivalent to
    \begin{equation*}
        Z \ce \sum_{n=1}^{\infty} \ee^{- \beta \gamma_n} < + \infty \ .
    \end{equation*}
    Observe that for all $x \in \R$, there holds\footnote{This is a modified version of an inequality stated in a similar context in Ref.~\cite[p. 396]{Lieb1972}.} $x \, \ee^{- \beta x} \le \frac{2}{\beta} \, \ee^{- \beta x / 2}$. In fact, from $x \le \ee^x$ it follows that $x \, \ee^{-x} \le 1$, and hence, by replacing $x$ with $\frac{\beta}{2} \, x$, also $\frac{\beta}{2} \, x \, \ee^{- \beta x / 2} \le 1$. Therefore, setting $\alpha_n \ce \ee^{- \beta \gamma_n} / Z$ ($n \in \N$), we can estimate
    \begin{align*}
        \sum_{n=1}^{\infty} \alpha_n \abs{\braket{e_n, T e_n}} = \frac{1}{Z} \sum_{n=1}^{\infty} \gamma_n \ee^{- \beta \gamma_n} \le \frac{2}{\beta Z} \sum_{n=1}^{\infty} \ee^{- \beta \gamma_n / 2} < + \infty \ .
    \end{align*}
    Defining the density operators $\rho_\beta \ce \sum_{n \in \N} \alpha_n P_{e_n}$, $\beta > 0$, the above shows that $\rho_\beta \in \DM(\MH, T)$; hence, the latter contains mixed states.
\end{proof}

\begin{remark}
    For a large class of potentials $V$ it can be shown that the self-adjoint Schrödinger operator $H = H_0 + V$, defined on a bounded region of $\R^n$, has the property that $\ee^{- \beta H}$ is a trace-class operator \cite[Appendix B, Thm. 8]{Lieb1972}, \cite[Thm. XIII.76]{RS4}. Therefore, the assumption of \cref{lem:existenceCompatibleDM} is justified.
\end{remark}

In the context of the set of admissible density operators defined in \cref{eq:admissibleDM}, it proves useful to observe the following property of the set of compatible density operators.

\begin{lemma}\label{lem:intersectionCompatibleDM}
    Let $M \in \N$ and $\set{T_i}_{i=1, \dotsc, M} \subset \LO(\MH)$ be self-adjoint operators on a Hilbert space $\MH$. Define the operator $T \ce \sum_{i=1}^{M} T_i$ on its natural domain $\dom(T) \ce \bigcap_{i=1}^{M} \dom(T_i)$ and assume that $T$ is self-adjoint. Then
    \begin{equation*}
        \bigcap_{i=1}^{M} \DM (\MH, T_i) \subset \DM (\MH, T) \quad \text{and} \quad \EV_\rho [T] = \sum_{i=1}^{M} \EV_\rho [T_i] \ .
    \end{equation*}
\end{lemma}

\begin{proof}
    Let $\rho \in \bigcap_{i=1}^{M} \DM (\MH, T_i)$ be arbitrary. Then there exist an orthonormal basis $(e_n)_{n \in \N}$ of $\MH$ and a sequence $(\alpha_n)_{n \in \N} \subset [0,1]$ with $\sum_{n \in \N} \alpha_n = 1$ such that (i) $e_n \in \bigcap_{i=1}^{M} \dom(T_i) = \dom(T)$, (ii) $\rho = \sum_{n \in \N} \alpha_n P_{e_n}$, and (iii)
    \begin{equation*}
        \sum_{n=1}^{\infty} \alpha_n \abs[\big]{\braket{e_n, T_i e_n}} < + \infty
    \end{equation*}
    for all $i \in \set{1, \dotsc, M}$. The first and the third property imply that
    \begin{align*}
        \sum_{n=1}^{\infty} \alpha_n \abs[\big]{\braket{e_n, T e_n}} = \sum_{n=1}^{\infty} \alpha_n \abs[\Bigg]{\,\sum_{i=1}^{M} \braket{e_n,  T_i e_n}} \le \sum_{i=1}^{M} \sum_{n=1}^{\infty} \alpha_n \abs[\big]{\braket{e_n, T_i e_n}} < + \infty \ ,
    \end{align*}
    where in the last step we were able to interchange the two summations because for each $i$ the sum over $n$ converges. Thus, it follows that $\rho \in \DM (\MH, T)$. The formula for the expectation value can be obtained by the very same reasoning.
\end{proof}

\subsection{Proof of the partial trace formula}\label[appsec]{subsec:proofPT}

In the following, we will provide the proof for \cref{lem:partialTrace} which asserts that for $T \in \LO(\MH_\SFS)$ self-adjoint and $\rho \in \DM(\MH_\SFS \otimes \MH_\SFR, T \otimes \id_\SFR)$, there holds $\tr_\SFR (\rho) \in \DM(\MH_\SFS, T)$ and
\begin{equation*}
    \tr_{\MH_\SFS \otimes \MH_\SFR} \bigl( (T \otimes \id_\SFR) \rho \bigr) = \tr_{\MH_\SFS} \bigl( T \tr_\SFR (\rho) \bigr) \ .
\end{equation*}

\begin{proof}[Proof of \cref{lem:partialTrace}]
    Let $(e_n)_{n \in \N} \subset \MH_\SFS$ and $(f_m)_{m \in \N} \subset \MH_\SFR$ be orthonormal bases and recall that $(e_n \otimes f_m)_{n,m \in \N}$ is an orthonormal basis of $\MH_\SFS \otimes \MH_\SFR$ \cite[p. 50]{RS1}. Assume that (i) $e_n \otimes f_m \in \dom(T \otimes \id_\SFR)$ for all $n, m \in \N$ (which is equivalent to $e_n \in \dom(T)$ for all $n \in \N$), (ii) $\rho = \sum_{n,m \in \N} \alpha_{n,m} P_{e_n \otimes f_m}$ with $\alpha_{n,m} \ge 0$ and $\sum_{n,m \in \N} \alpha_{n,m} = 1$, and (iii)
    \begin{equation*}
        \sum_{n,m \in \N} \alpha_{n,m} \abs[\Big]{\braket[\big]{e_n \otimes f_m, (T \otimes \id_\SFR) (e_n \otimes f_m)}_{\MH_\SFT}} < + \infty \ .
    \end{equation*}
    
    For any $\eta \in \MH_\SFR$ let $U_\eta : \MH_\SFS \to \MH_\SFS \otimes \MH_\SFR$ be given by $U_\eta \xi \ce \xi \otimes \eta$, $\xi \in \MH_\SFS$. Then it holds that the partial trace $\rho_\SFS \ce \tr_\SFR (\rho)$ of $\rho$ can be written as \cite[Eq. (14.5)]{vanNeerven2022}
    \begin{equation*}
        \rho_\SFS = \sum_{m=1}^{\infty} U_{f_m}^\ast \rho U_{f_m} \ ,
    \end{equation*}
    where the adjoint $U_\eta^\ast : \MH_\SFS \otimes \MH_\SFR \to \MH_\SFS$ acts like $U_\eta^\ast (\xi \otimes \eta^\prime) = \braket{\eta, \eta^\prime}_{\MH_\SFR} \xi$ for $\xi \in \MH_\SFS$, $\eta, \eta^\prime \in \MH_\SFR$. Thus, given the spectral decomposition of $\rho$, the partial trace $\rho_\SFS$ satisfies
    \begin{align*}
        \rho_\SFS \xi &= \sum_{m=1}^{\infty} U_{f_m}^\ast \rho (\xi \otimes f_m) = \sum_{m=1}^{\infty} \sum_{n=1}^{\infty} \sum_{r=1}^{\infty} \alpha_{n, r} \braket{e_n \otimes f_r, \xi \otimes f_m}_{\MH_\SFT} U_{f_m}^\ast (e_n \otimes f_r) \\
        &= \sum_{m=1}^{\infty} \sum_{n=1}^{\infty} \sum_{r=1}^{\infty} \alpha_{n, r} \braket{e_n, \xi}_{\MH_\SFS} \braket{f_r, f_m}_{\MH_\SFR} \braket{f_m, f_r}_{\MH_\SFR} e_n = \sum_{n=1}^{\infty} \sum_{m=1}^{\infty} \alpha_{n,m} \braket{e_n, \xi}_{\MH_\SFS} e_n \ .
    \end{align*}
    Note that in the last step, we were able to interchange the two sums according to Fubini's theorem because the double series converges absolutely. Hence, setting $\alpha_n \ce \sum_{m \in \N} \alpha_{n,m}$ for all $n \in \N$ (and noting that $\sum_{n \in \N} \alpha_n = 1$), it follows that $\rho_\SFS$ is given by
    \begin{equation}\label{eq:spectralReprPT}
       \rho_\SFS = \sum_{n=1}^{\infty} \alpha_n P_{e_n} \ .
    \end{equation}
    Observe that by assumption (iii) from above, we have
    \begin{align*}
        \sum_{n=1}^{\infty} \alpha_n \abs{\braket{e_n, T e_n}_{\MH_\SFS}} &= \sum_{n=1}^{\infty} \sum_{m=1}^{\infty} \alpha_{n,m} \abs[\big]{\braket{e_n, T e_n}_{\MH_\SFS} \braket{f_m, f_m}_{\MH_\SFR}} \\
        &= \sum_{n=1}^{\infty} \sum_{m=1}^{\infty} \alpha_{n,m} \abs[\big]{\braket[\big]{e_n \otimes f_m, (T \otimes \id_\SFR) (e_n \otimes f_m)}_{\MH_\SFT}} < + \infty \ .
    \end{align*}
    Since $e_n \in \dom(T)$ as noted above, this shows that $\rho_\SFS = \tr_\SFR (\rho) \in \DM(\MH_\SFS, T)$, which proves the first claim.

    To verify \cref{eq:partialTrace}, we can use the defintion \eqref{eq:trace} of the trace to rewrite the left-hand side of \eqref{eq:partialTrace} and then simply calculate:
    \begin{align*}
        \tr_{\MH_\SFS \otimes \MH_\SFR} \bigl( (T \otimes \id_\SFR) \rho \bigr) &= \sum_{n=1}^{\infty} \sum_{m=1}^{\infty} \alpha_{n,m} \braket[\big]{e_n \otimes f_m, (T \otimes \id_\SFR) (e_n \otimes f_m)}_{\MH_\SFT} \\
        &= \sum_{n=1}^{\infty} \sum_{m=1}^{\infty} \alpha_{n,m} \braket{e_n, T e_n}_{\MH_\SFS} \braket{f_m, f_m}_{\MH_\SFR} = \sum_{n=1}^{\infty} \alpha_n \braket{e_n, T e_n}_{\MH_\SFS} \\
        &= \tr_{\MH_\SFS} \bigl( T \rho_\SFS \bigr) \ ,
    \end{align*}
    where the last step used the definition of the quantity $\tr_{\MH_\SFS} \bigl( T \rho_\SFS \bigr)$ according to \cref{eq:trace} together with the representation \eqref{eq:spectralReprPT} and the fact that $\rho_\SFS \in \DM(\MH_\SFS, T)$.
\end{proof}

\section{Background from convex geometry}\label[appsec]{app:convexGeometry}

Let $d \in \N$ be arbitrary. Denote by $B^d \ce \set{x \in \R^d \, : \, \abs{x} < 1} \subset \R^d$ the $d$-dimensional Euclidean unit ball centered at the origin, and let $\omega_d \ce \vol(B^d)$ be its volume. For two sets $A, B \subset \R^d$ and $\alpha, \beta \in \R$, their Minkowski linear combination is defined as
\begin{equation*}
    \alpha A + \beta B \ce \set[\big]{\alpha x + \beta y \, : \, x \in A, \, y \in B} \ .
\end{equation*}
It holds that $A + B = \bigcup_{a \in A} (a + B) = \bigcup_{b \in B} (A + b)$ \cite[p. 5]{HugWeil2020}. Moreover, it is evident that for all $\varepsilon > 0$ and $a \in A$ the set $\varepsilon B^d$ is given by the ball $B(0, \varepsilon)$ of radius $\varepsilon$, and the set $a + \varepsilon B^d = \set{a + x \, : \, \abs{x} < \varepsilon}$ is equal to the $\varepsilon$-ball $B(a, \varepsilon)$ centered at $a$. 

\begin{definition}\label{def:epsilonNeighborhood}
    Let $A \subset \R^d$ be any set and $\varepsilon > 0$ be a positive number. The \emph{tubular $\varepsilon$-neighborhood} (or: $\varepsilon$-thickening) of $A$ is defined to be the set
    \begin{equation*}
        A_\varepsilon \ce \set[\big]{x \in \R^d \, : \, \dist(x, A) < \varepsilon} \ ,
    \end{equation*}
    where $\dist(x, A) \ce \inf_{a \in A} \abs{x - a}$ is the distance between the point $x$ and the set $A$. That is, $A_\varepsilon$ consists of all those points in $\R^d$ whose distance to $A$ is smaller than $\varepsilon$.
\end{definition}

\begin{remark}
    It is straightforward to show that the $\varepsilon$-neighborhood of $A$ can also be written as $A_\varepsilon = \bigcup_{a \in A} B(a, \varepsilon)$. Therefore, it readily follows that the Minkowski sum $A + \varepsilon B^d$ is equal to the $\varepsilon$-neighborhood of $A$:
    \begin{equation*}
        A + \varepsilon B^d = \bigcup_{a \in A} (a + \varepsilon B^d) = \bigcup_{a \in A} B(a, \varepsilon) = A_\varepsilon \ .
    \end{equation*}
\end{remark}

A cornerstone in convex geometry is the \emph{Minkowski-Steiner formula}, which expresses the volume of the set $A + \varepsilon B^d = A_\varepsilon$, where $A \subset \R^d$ is convex and $\varepsilon > 0$, as a polynomial in the parameter $\varepsilon$. The proof of this result can be found, e.g., in Refs.~\cite[p. 141]{BuragoZalgaller1988}, \cite[Thm. 6.6]{Gruber2007}, and \cite[Thm. 3.10]{HugWeil2020} (see also Ref.~\cite{Trudinger1997} and references therein).

\begin{theorem}[Minkowski-Steiner]\label{thm:MinkowskiSteiner}
    Let $A \subset \R^d$ be an open, bounded, convex set and $\varepsilon > 0$ be arbitrary. Then there are numbers $W_2(A), \dotsc, W_{d-1} (A) \ge 0$ such that
    \begin{equation}\label{eq:MinkowskiSteiner}
        \vol \bigl(A + \varepsilon B^d\bigr) = \vol (A) + \area (\partial A) \, \varepsilon + \sum_{j=2}^{d-1} \binom{d}{j} W_j (A) \, \varepsilon^j + \omega_d \, \varepsilon^d \ .
    \end{equation}
\end{theorem}

\begin{remarks}\label{rem:MinkowskiSteiner}
    \leavevmode
    \begin{enumerate}[env]
        \item The coefficients $W_j(A)$ in the above polynomial expansion are called \emph{quermassintegrals} of the set $A$. They are defined for all $j \in \set{0, 1, \dotsc, d}$ by
        \begin{equation}\label{eq:quermassintegral}
            W_j(A) \ce V \bigl( \, \underbrace{A, \dotsc, A}_{d - j}, \, \underbrace{B^d, \dotsc, B^d}_{j} \, \bigr) \ ,
        \end{equation}
        with $V (A, \dotsc, A, B^d, \dotsc, B^d)$ being the \emph{mixed volume} of the convex sets $A, \dotsc, A, B^d, \dotsc, B^d$ \cite[Thm. 6.5]{Gruber2007}. The function $V$ is symmetric, linear, continuous with respect to the Hausdorff metric, and non-decreasing in each argument \cite[Sec. 6.3]{Gruber2007}, \cite[Thm. 3.9]{HugWeil2020}.
        
        \item In writing \cref{eq:MinkowskiSteiner}, we have used the following two identities: (i) $W_0 (A) = \vol(A)$ is the volume of $A$,  and (ii) $W_1(A) = \frac{1}{d} \, \area(\partial A)$ is proportional to the surface area of the boundary of $A$ \cite[p. 139]{BuragoZalgaller1988}, \cite[Thm. 3.2.35]{Federer1996}, \cite[Thm. 3.9 (a)]{HugWeil2020}, \cite{Trudinger1997}.
    \end{enumerate}
\end{remarks}

\begin{example}\label{exa:MinkowskiSteiner3D}
    Consider the specific case $d = 3$. Aside from the three terms that are already stated explicitly in the Minkowski-Steiner formula \eqref{eq:MinkowskiSteiner}, the remaining sum only contributes the $j = 2$ term. The corresponding quermassintegral is given by \cite[Eq. (3.16)]{HugWeil2020}
    \begin{equation*}
        W_2 (K) = V (K, B_3, B_3) = \frac{1}{3} \, \int_{S^2} h_K \ ,
    \end{equation*}
    where $h_K(u) \ce \sup_{x \in K} \braket{u, x}$, $u \in S^2$, is the support function of the convex set $K$, $S^2$ denotes the 2-sphere in $\R^3$, and integration is with respect to the two-dimensional Hausdorff measure. With this identity, the Minkowski-Steiner formula reduces to
    \begin{equation*}
        \vol (K_\varepsilon) = \vol (K) + \area (\partial K) \, \varepsilon + \int_{S^2} h_K \, \varepsilon^2 + \omega_3 \, \varepsilon^3 \ .
    \end{equation*}
    Suppose that $0 \in K$. Then $\sup_{x \in K} \norm{x} \le \diam(K)$, where $\diam(K) \ce \sup_{x, y \in K} \norm{x - y}$ denotes the diameter of $K$. This implies $h_K(u) \le \diam(K)$ for all $u \in S^2$, using the Cauchy-Schwarz inequality. Since $\area(S^2) = 3 \omega_3$, it follows that
    \begin{equation*}
        \vol (K_\varepsilon \setminus K) \le \area (\partial K) \, \varepsilon + 3 \omega_3 \, \diam(K) \, \varepsilon^2 + \omega_3 \, \varepsilon^3 \ .
    \end{equation*}
\end{example}

The second result that we need is an estimate of the volume of the inner tubular neighborhood of a set; it is taken from Ref.~\cite[Lem. 2.2]{Farrington2025} (see also Ref.~\cite[Rem. 5.7]{Gittins2020}).

\begin{lemma}[{\protect\cite[Lem. 2.2]{Farrington2025}}]\label{lem:boundThinning}
    Let $A \subset \R^d$ be open, bounded, and convex. For all $\varepsilon > 0$,
    \begin{equation*}
        \vol \bigl(\set{x \in A \, : \, \dist(x, \partial A) \le \varepsilon}\bigr) \le \area (\partial A) \, \varepsilon \ .
    \end{equation*}
\end{lemma}

\section{Characterization of the Hilbert space direct sum}\label[appsec]{app:directSum}

Let $I$ be an at most countable index set, and for every $i \in I$ let $(\MH_i, \bdot_{\MH_i})$ be a Hilbert space. Consider the following subset of the direct product $\prod_{i \in I} \MH_i$:
\begin{equation}\label{eq:directSum}
    \bigoplus_{i \in I} \MH_i \ce \Set{\Psi = \bigl( \psi^{(i)} \bigr)_{i \in I} \in \prod_{i \in I} \MH_i \ : \ \sum_{i \in I} \norm[\big]{\psi^{(i)}}_{\MH_i}^2 < + \infty} \ .
\end{equation}
It is straightforward to show that $\bigoplus_{i \in I} \MH_i$ is a linear subspace of $\prod_{i \in I} \MH_i$, with algebraic operations defined component-wise. On this subspace, define an inner product $\bdot_\oplus$ by
\begin{equation*}
    \braket{\Psi, \Phi}_\oplus \ce \sum_{i \in I} \braket[\Big]{\psi^{(i)}, \phi^{(i)}}_{\MH_i} \ .
\end{equation*}
The Cauchy-Schwarz and Hölder inequality show that $\bdot_\oplus$ is well-defined. Moreover, one can show by standard arguments that $\bigoplus_{i \in I} \MH_i$ is complete with respect to the norm
\begin{equation}\label{eq:directSumNorm}
    \norm{\Psi}_\oplus^2 \ce \sum_{i \in I} \norm[\big]{\psi^{(i)}}_{\MH_i}^2
\end{equation}
induced by the inner product \cite[Thm. 18.1]{BlanchardBruening2015}. Therefore, $\bigl( \bigoplus_{i \in I} \MH_i, \bdot_\oplus \bigr)$ is a Hilbert space called the \emph{Hilbert space direct sum} of the family $(\MH_i)_{i \in I}$.

In the recent paper \cite{Fritz2020}, infinite direct sums in arbitrary $C^\ast$-categories were characterized in terms of a universal property; the following theorem is a very special case of their much more general main result. Note that we only state one of the two equivalent characterizations given there, which corresponds to the definition of an infinite direct sum in $W^\ast$-categories first given in Ref.~\cite[p. 100]{Ghez1985}.

\begin{theorem}[{\protect\cite[Thm. 5.1, (a) $\Leftrightarrow$ (c)]{Fritz2020}}]\label{thm:directSum}
    Let $(\MH_i)_{i \in I}$ be a family of Hilbert spaces and $\MH$ be another Hilbert space. The following assertions are equivalent:
    \begin{enumerate}[labelindent=\parindent, leftmargin=*, label=\normalfont(\roman*), align=right]
        \item $\MH$ is the direct sum $\bigoplus_{i \in I} \MH_i$;
        \item there exists a family of bounded linear operators $(A_i : \MH_i \to \MH)_{i \in I}$ such that
        \begin{equation*}
            A_i^\ast A_j = \delta_{ij} \id_{\MH_i} \quad \text{and} \quad \sum_{i \in I} A_i A_i^\ast = \id_\MH \ .
        \end{equation*}
    \end{enumerate}
\end{theorem}

With the help of this theorem, we can give the following alternative proof of \cref{pro:FockSpace}, which is a bit longer than the one given above using only the spectral theorem, but also structurally more interesting as it connects the physical assumptions \ref{enu:A6.1} and \ref{enu:A6.2} with the mathematically fundamental universal property of the direct sum.

\begin{proof}[Proof of \cref{pro:FockSpace}]
    For $n \in \N_0$ let $i_n : \MH^{\otimes n} \hookrightarrow \MH_\SFS$, $\xi \mapsto i_n(\xi)$, denote the inclusion mapping, which is an isometry according to Assumption \ref{enu:A6.1}: $\norm{i_n(\xi)}_{\MH_\SFS} = \norm{\xi}_{\MH^{\otimes n}}$ Polarization gives $\braket{i_n(\xi), i_n(\eta)}_{\MH_\SFS} = \braket{\xi, \eta}_{\MH^{\otimes n}}$, and this implies $i_n^\ast i_n = \id_{\MH^{\otimes n}}$ \cite[Prop. II.2.17]{Conway1990}. Moreover, for $n, m \in \N$ with $n \neq m$ and $\xi \in \MH^{\otimes n}$, $\eta \in \MH^{\otimes m}$, we have
    \begin{equation*}
        \braket[\big]{\xi, i_n^\ast i_m (\eta)}_{\MH^{\otimes n}} = \braket[\big]{i_n(\xi), i_m(\eta)}_{\MH_\SFS} = 0
    \end{equation*}
    because $\MH^{\otimes n} = \Eig(\MSN, n) \perp \Eig(\MSN, m) = \MH^{\otimes m}$ by Assumption \ref{enu:A6.2}. This implies that $i_n^\ast i_m = 0$ for $n \neq m$; hence, for all $n, m \in \N_0$, there holds
    \begin{equation}\label{eq:iotaFirstProperty}
        i_n^\ast i_m = \delta_{nm} \id_{\MH^{\otimes n}} \ .
    \end{equation}

    Define $P_n \ce i_n i_n^\ast \in \BO(\MH_\SFS)$, $n \in \N_0$, and note that $P_n^\ast = P_n$ and $P_n^2 = P_n$, showing that the $P_n$ are orthogonal projections. Let $\xi \in \MH^{\otimes n}$ be arbitrary. Then $P_n i_n(\xi) = i_n i_n^\ast i_n(\xi) = i_n(\xi)$, so $i_n(\xi) \in \ran(P_n)$ and $\MH^{\otimes n} \subset \ran(P_n)$. Next, let $\Phi \in (\MH^{\otimes n})^\perp \subset \MH_\SFS$, meaning $\braket{\Phi, i_n(\xi)}_{\MH_\SFS} = 0$ for all $\xi \in \MH^{\otimes n}$. If $\Psi \in \MH_\SFS$, it follows that
    \begin{equation*}
        \braket[\big]{\Psi, P_n \Phi}_{\MH_\SFS} = \braket[\big]{i_n i_n^\ast (\Psi), \Phi}_{\MH_\SFS} = 0
    \end{equation*}
    and thus $\Phi \in \ker(P_n)$, showing that $(\MH^{\otimes n})^\perp \subset \ker(P_n)$. Since $P_n$ is an orthogonal projection, $\ker(P_n) = \ran(P_n)^\perp$ and so $(\MH^{\otimes n})^\perp \subset \ran(P_n)^\perp$. This implies that if $\Psi \in \ran(P_n)$, then $\Psi \perp \Phi$ for all $\Phi \in (\MH^{\otimes n})^\perp$, and so $\Psi \in (\MH^{\otimes n})^{\perp\perp} = \oln{\MH^{\otimes n}} = \MH^{\otimes n}$ since $\MH^{\otimes n}$ is closed; hence, $\ran(P_n) \subset \MH^{\otimes n}$. In total, we have $\ran(P_n) = \MH^{\otimes n}$, so $P_n$ must be the unique orthogonal projection onto the closed subspace $\MH^{\otimes n} \subset \MH_\SFS$.

    Let $E_\MSN$ denote the spectral measure of the self-adjoint operator $\MSN$. It holds that $E_\MSN(\set{n})$, $n \in \N_0$, is the orthogonal projection in $\MH_\SFS$ onto $\Eig(\MSN, n) = \MH^{\otimes n}$ \cite[Prop. 5.10]{Schmuedgen2012}; thus, the previous result implies that $E_\MSN(\set{n}) = P_n = i_n i_n^\ast$. From countable additivity of the spectral measure, we derive
    \begin{equation}\label{eq:iotaSecondProperty}
        \sum_{n=0}^{\infty} i_n i_n^\ast = \sum_{n=0}^{\infty} E_\MSN \bigl( \set{n} \bigr) = E_\MSN \biggl(\, \bigcup_{n=0}^{\infty} \set{n} \biggr) = E_\MSN \bigl( \sigma(\MSN) \bigr) = \id_{\MH_\SFS} \ .
    \end{equation}

    In light of \cref{eq:iotaFirstProperty,eq:iotaSecondProperty}, we may invoke the universal property of the Hilbert space direct sum from \cref{thm:directSum} to conclude that $\MH_S$ must be the direct sum of the spaces $\MH^{\otimes n}$, which proves the assertion of \cref{pro:FockSpace}.
\end{proof}

\section{Functional calculus for tensor product operators}\label[appsec]{app:functionalCalculusTP}

In the following, we prove a formula for the functional calculus of tensor product operators that is used in the proof of \cref{cor:vonNeumannEquation}.

\begin{lemma}\label{lem:functionalCalculusTP}
    Let $\MH_1$ and $\MH_2$ be Hilbert spaces, $T \in \LO(\MH_1)$ be a self-adjoint operator on $\MH_1$, and $f : \R \to \C$ be a bounded function. Then
    \begin{equation*}
        f (T \otimes \id_{\MH_2}) = f(T) \otimes \id_{\MH_2} \ .
    \end{equation*}
\end{lemma}

\begin{proof}
    The two key ingredients for the proof are the facts that (i) the spectral measure $E_{T \otimes \id_{\MH_2}}$ of the self-adjoint operator $T \otimes \id_{\MH_2}$ is given by $E_{T \otimes \id_{\MH_2}} (B) = E_{T} (B) \otimes \id_{\MH_2}$ for all Borel subsets $B \subset \R$ \cite[Thm. 3.5 (ii)]{Arai2024} (where $E_T$ is the spectral measure of $T$), and that (ii) the spectrum of $T \otimes \id_{\MH_2}$ is given by $\sigma (T \otimes \id_{\MH_2}) = \sigma(T)$ \cite[Thm. 3.6]{Arai2024}.

    Let $\xi \in \MH_1$ and $\eta \in \MH_2$ be arbitrary. From the standard properties of the bounded functional calculus, it follows that
    \begin{align*}
        \braket[\big]{\xi \otimes \eta, f (T \otimes \id_{\MH_2}) (\xi \otimes \eta)} &= \int_{\sigma (T \otimes \id_{\MH_2})} f(\lambda) \diff \braket[\big]{\xi \otimes \eta, E_{T \otimes \id_{\MH_2}} (\lambda) (\xi \otimes \eta)} \\
        &= \int_{\sigma (T \otimes \id_{\MH_2})} f(\lambda) \diff \braket[\big]{\xi \otimes \eta, (E_T (\lambda) \otimes \id_{\MH_2}) (\xi \otimes \eta)} \\
        &= \int_{\sigma (T)} f(\lambda) \diff \braket{\xi, E_T (\lambda) \xi} \cdot \braket{\eta, \id_{\MH_2} \eta} \\
        &= \braket[\big]{\xi \otimes \eta, (f(T) \otimes \id_{\MH_2}) (\xi \otimes \eta)} \ .
    \end{align*}
    Since the linear span of elementary tensors $\xi \otimes \eta$ (that is, the algebraic tensor product $\MH_1 \odot \MH_2$) lies dense in $\MH_1 \otimes \MH_2$, we conclude that $f (T \otimes \id_{\MH_2}) = f(T) \otimes \id_{\MH_2}$.
\end{proof}

\begin{remark}
    In the setting of the lemma, suppose that $f$ is an $E_T$-almost everywhere finite function and that $\xi \in \dom \bigl( f(T) \bigr)$. Then the above proof shows that $\xi \otimes \eta \in \dom \bigl( f(T \otimes \id_{\MH_2}) \bigr)$ and $f (T \otimes \id_{\MH_2}) = f(T) \otimes \id_{\MH_2}$ on $\dom \bigl( f(T) \otimes \id_{\MH_2} \bigr)$ because $\dom \bigl( f(T) \bigr) \odot \MH_2$ is dense in the latter. Therefore, $f(T) \otimes \id_{\MH_2} \subset f (T \otimes \id_{\MH_2})$.
\end{remark}

\section*{Acknowledgments}

This work was supported by the DFG Collaborative Research Center 1114 \enquote{Scaling Cascades in Complex Systems}, Project No.~235221301, Project C01 \enquote{Adaptive coupling of scales in molecular dynamics and beyond to fluid dynamics}, and by the DFG, Project No.~DE 1140/15-1, \enquote{Mathematical model and numerical implementation of open quantum systems in molecular simulation}.

\bibliographystyle{abbrv}
\bibliography{hamiltonian}

\begin{thebibliography}{10}

\bibitem{Ali2020}
M.~M. Ali, W.-M. Huang, and W.-M. Zhang.
\newblock Quantum thermodynamics of single particle systems.
\newblock {\em Sci. Rep.}, 10:13500, 2020.

\bibitem{Arai2024}
A.~Arai.
\newblock {\em Analysis on Fock Spaces and Mathematical Theory of Quantum Fields}.
\newblock World Scientific, Singapore, {S}econd edition, 2024.

\bibitem{AttalJoyePillet2006}
S.~Attal, A.~Joye, and C.-A. Pillet, editors.
\newblock {\em Open Quantum Systems}.
\newblock Lecture Notes in Mathematics 1880--1882. Springer, Berlin, Heidelberg, 2006.
\newblock 3 Volumes. Vol. I: The Hamiltonian Approach, Vol. II: The Markovian Approach, Vol. III: Recent Developments.

\bibitem{Blakaj2025}
V.~Blakaj, M.~C. Caro, A.~Kouraich, D.~Malz, and M.~M. Wolf.
\newblock Gibbs state postulate from dynamical stability -- redundancy of the zeroth law.
\newblock {\em Preprint}, 2025.
\newblock arXiv:2512.13451 [math-ph].

\bibitem{Blanchard2011}
P.~Blanchard and E.~Brüning.
\newblock Reply to “{C}omment on ‘{R}emarks on the structure of states of composite quantum systems and envariance’ [{P}hys. {L}ett. {A} 355 (2006) 180]” [{P}hys. {L}ett. {A} 375 (2011) 1160].
\newblock {\em Phys. Lett. A}, 375(7):1163--1165, 2011.

\bibitem{BlanchardBruening2015}
P.~Blanchard and E.~Brüning.
\newblock {\em Mathematical Methods in Physics}.
\newblock Progress in Mathematical Physics 69. Birkhäuser, Cham, Heidelberg, New York, {S}econd edition, 2015.

\bibitem{Bogoliubov1958}
N.~N. Bogoliubov.
\newblock On a new method in the theory of superconductivity.
\newblock {\em Il Nuovo Cimento}, 7:794–805, 1958.

\bibitem{BratteliRobinson1981}
O.~Bratteli and D.~W. Robinson.
\newblock {\em Operator Algebras and Quantum Statistical Mechanics II: Equilibrium States, Models in Quantum Statistical Mechanics}.
\newblock Theoretical and Mathematical Physics. Springer, Berlin, Heidelberg, New York, 1981.

\bibitem{BreuerPetruccione2007}
H.-P. Breuer and F.~Petruccione.
\newblock {\em The Theory of Open Quantum Systems}.
\newblock Oxford University Press, Oxford, UK, 2007.

\bibitem{Brydges1999}
D.~C. Brydges and P.~A. Martin.
\newblock Coulomb systems at low density: A review.
\newblock {\em J. Stat. Phys.}, 96:1163–1330, 1999.

\bibitem{BuragoZalgaller1988}
Y.~D. Burago and V.~A. Zalgaller.
\newblock {\em Geometric Inequalities}.
\newblock Number 285 in Grundlehren der mathematischen Wissenschaften. Springer, Berlin, Heidelberg, 1988.

\bibitem{Chaiken1967}
J.~M. Chaiken.
\newblock Finite-particle representations and states of the canonical commutation relations.
\newblock {\em Ann. Phys.}, 42(1):23--80, 1967.

\bibitem{Chaiken1968}
J.~M. Chaiken.
\newblock Number operators for representations of the canonical commutation relations.
\newblock {\em Commun. Math. Phys.}, 8:164–184, 1968.

\bibitem{Conway1990}
J.~B. Conway.
\newblock {\em A Course in Functional Analysis}.
\newblock Graduate Texts in Mathematics 96. Springer, New York, NY, {S}econd edition, 1990.

\bibitem{Cook1953}
J.~M. Cook.
\newblock The mathematics of second quantization.
\newblock {\em Trans. Amer. Math. Soc.}, 74(2):222--245, 1953.

\bibitem{Courbage1971}
M.~Courbage, S.~{Miracle-Sole}, and D.~W. Robinson.
\newblock Normal states and representations of the canonical commutation relations.
\newblock {\em Ann. Inst. Henri Poincaré. Sect. A, Physique Théorique}, 14(2):171--178, 1971.

\bibitem{Davies1976}
E.~B. Davies.
\newblock {\em Quantum Theory of Open Systems}.
\newblock Academic Press, London, New York, 1976.

\bibitem{delRazo2025}
M.~J. {del Razo} and L.~{Delle Site}.
\newblock Dynamics of systems with varying number of particles: From {L}iouville equations to general master equations for open systems.
\newblock {\em SciPost Phys.}, 18:001, 2025.

\bibitem{delRazo2022}
M.~J. {del Razo}, D.~Frömberg, A.~V. Straube, et~al.
\newblock A probabilistic framework for particle-based reaction–diffusion dynamics using classical fock space representations.
\newblock {\em Lett. Math. Phys.}, 112:49, 2022.

\bibitem{DellAntonio1967}
G.-F. {Dell'Antonio} and S.~Doplicher.
\newblock Total number of particles and {F}ock representation.
\newblock {\em J. Math. Phys.}, 8(3):663–666, 1967.

\bibitem{DellAntonio1966}
G.-F. {Dell'Antonio}, S.~Doplicher, and D.~Ruelle.
\newblock A theorem on canonical commutation and anticommutation relations.
\newblock {\em Commun. Math. Phys.}, 2:223–230, 1966.

\bibitem{DelleSite2018}
L.~{Delle Site}.
\newblock Simulation of many-electron systems that exchange matter with the environment.
\newblock {\em Adv. Theory Simul.}, 1(10):1800056, 2018.

\bibitem{DelleSite2024jpa}
L.~{Delle Site} and A.~Djurdjevac.
\newblock An effective {Hamiltonian} for the simulation of open quantum molecular systems.
\newblock {\em J. Phys. A: Math. Theor.}, 57(25):255002, 2024.

\bibitem{DelleSite2020}
L.~{Delle Site} and R.~Klein.
\newblock Liouville-type equation for the $n$-particle distribution function of an open system.
\newblock {\em J. Math. Phys.}, 61(8):083102, 2020.

\bibitem{Derezinski2026}
J.~Derezi\'{n}ski, V.~Jak\v{s}i\'{c}, and C.-A. Pillet.
\newblock Miniatures on open quantum systems.
\newblock {\em Preprint}, 2026.
\newblock arXiv:2601.20373 [math-ph].

\bibitem{Farrington2025}
S.~Farrington.
\newblock On the isoperimetric and isodiametric inequalities and the minimisation of eigenvalues of the {L}aplacian.
\newblock {\em J. Geom. Anal.}, 35:62, 2025.

\bibitem{Federer1996}
H.~Federer.
\newblock {\em Geometric Measure Theory}.
\newblock Classics in Mathematics. Springer, Berlin, Heidelberg, 1996.

\bibitem{Fock1932}
V.~Fock.
\newblock Konfigurationsraum und zweite {Q}uantelung.
\newblock {\em Z. Phys.}, 75:622--647, 1932.

\bibitem{Friedrichs1952}
K.~O. Friedrichs.
\newblock Mathematical aspects of the quantum theory of fields. {P}art {IV}. {O}ccupation number representation and fields of different kinds.
\newblock {\em Commun. Pure Appl. Math.}, 5(4):349--411, 1952.

\bibitem{Fritz2020}
T.~Fritz and B.~Westerbaan.
\newblock The universal property of infinite direct sums in {C}$^\ast$-categories and {W}$^\ast$-categories.
\newblock {\em Appl. Categor. Struct.}, 28:355–365, 2020.

\bibitem{Ghez1985}
P.~Ghez, R.~Lima, and J.~E. Roberts.
\newblock {$W^\ast$}-categories.
\newblock {\em Pacific J. Math.}, 120(1):79–109, 1985.

\bibitem{Gittins2020}
K.~Gittins and C.~Léna.
\newblock Upper bounds for {C}ourant-sharp {N}eumann and {R}obin eigenvalues.
\newblock {\em Bull. Soc. Math. France}, 148(1):99--132, 2020.

\bibitem{Gruber2007}
P.~M. Gruber.
\newblock {\em Convex and Discrete Geometry}.
\newblock Grundlehren der mathematischen Wissenschaften 336. Springer, Berlin, Heidelberg, 2007.

\bibitem{Haag1996}
R.~Haag.
\newblock {\em Local Quantum Physics}.
\newblock Theoretical and Mathematical Physics. Springer, Berlin, Heidelberg, New York, {S}econd {R}evised and {E}nlarged edition, 1996.

\bibitem{Hainzl2009I}
C.~Hainzl, M.~Lewin, and J.~P. Solovej.
\newblock The thermodynamic limit of quantum {C}oulomb systems {P}art {I}. {G}eneral theory.
\newblock {\em Adv. Math.}, 221(2):454--487, 2009.

\bibitem{Hainzl2009II}
C.~Hainzl, M.~Lewin, and J.~P. Solovej.
\newblock The thermodynamic limit of quantum {C}oulomb systems {P}art {II}. {A}pplications.
\newblock {\em Adv. Math.}, 221(2):488--546, 2009.

\bibitem{Huang1987}
K.~Huang.
\newblock {\em Statistical Mechanics}.
\newblock John Wiley \& Sons, New York, NY, {S}econd edition, 1987.

\bibitem{HugWeil2020}
D.~Hug and W.~Weil.
\newblock {\em Lectures on Convex Geometry}.
\newblock Graduate Texts in Mathematics 286. Springer, Cham, 2020.

\bibitem{Klein2022}
R.~Klein and L.~{Delle Site}.
\newblock Derivation of {Liouville}-like equations for the $n$-state probability density of an open system with thermalized particle reservoirs and its link to molecular simulation.
\newblock {\em J. Phys. A: Math. Theor.}, 55(15):155002, 2022.

\bibitem{Lieb1972}
E.~H. Lieb and J.~L. Lebowitz.
\newblock The constitution of matter: Existence of thermodynamics for systems composed of electrons and nuclei.
\newblock {\em Adv. Math.}, 9(3):316--398, 1972.
\newblock With an appendix by B. Simon.

\bibitem{LiebSeiringer2010}
E.~H. Lieb and R.~Seiringer.
\newblock {\em The Stability of Matter in Quantum Mechanics}.
\newblock Cambridge University Press, Cambridge, UK, 2010.

\bibitem{LandauLifshitz9}
E.~M. Lifshitz and L.~P. Pitaevskii.
\newblock {\em Statistical Physics, Part 2: Theory of the Condensed State}.
\newblock Vol. 9 of Landau and Lifshitz Course of Theoretical Physics. Pergamon Press, Oxford, UK, 1980.

\bibitem{Manzano2020}
D.~Manzano.
\newblock A short introduction to the {Lindblad} master equation.
\newblock {\em AIP Advances}, 10(2):025106, 2020.

\bibitem{Merkli2020}
M.~Merkli.
\newblock Quantum {M}arkovian master equations: Resonance theory shows validity for all time scales.
\newblock {\em Ann. Phys.}, 412:167996, 2020.

\bibitem{Mishin2015}
Y.~Mishin.
\newblock Thermodynamic theory of equilibrium fluctuations.
\newblock {\em Ann. Phys.}, 363:48--97, 2015.

\bibitem{Moretti2017}
V.~Moretti.
\newblock {\em Spectral Theory and Quantum Mechanics}.
\newblock Springer, Cham, {S}econd edition, 2017.

\bibitem{Moretti2019}
V.~Moretti.
\newblock {\em Fundamental Mathematical Structures of Quantum Theory}.
\newblock Springer, Cham, 2019.

\bibitem{Nolting2018}
W.~Nolting.
\newblock {\em Theoretical Physics 8. Statistical Physics}.
\newblock Springer, Cham, 2018.

\bibitem{Osserman1978}
R.~Osserman.
\newblock The isoperimetric inequality.
\newblock {\em Bull. Amer. Math. Soc.}, 84(6):1182--1238, 1978.

\bibitem{RS2}
M.~Reed and B.~Simon.
\newblock {\em Methods of Modern Mathematical Physics. Volume II: Fourier Analysis, Self-Adjointness}.
\newblock Academic Press, New York, San Francisco, London, 1975.

\bibitem{RS4}
M.~Reed and B.~Simon.
\newblock {\em Methods of Modern Mathematical Physics. Volume IV: Analysis of Operators}.
\newblock Academic Press, New York, San Francisco, London, 1978.

\bibitem{RS1}
M.~Reed and B.~Simon.
\newblock {\em Methods of Modern Mathematical Physics. Volume I: Functional Analysis}.
\newblock Academic Press, San Diego, New York, Boston, {R}evised and {E}nlarged edition, 1980.

\bibitem{Reible2025pre}
B.~M. Reible and L.~{Delle Site}.
\newblock Open quantum systems and the grand canonical ensemble.
\newblock {\em Phys. Rev. E}, 112(2):024130, 2025.

\bibitem{Reible2025apq}
B.~M. {Reible}, A.~Djurdjevac, and L.~{Delle Site}.
\newblock Chemical potential and variable number of particles control the quantum state: Quantum oscillators as a showcase.
\newblock {\em APL Quantum}, 2(1):016124, 2025.

\bibitem{Reible2022}
B.~M. Reible, C.~Hartmann, and L.~{Delle Site}.
\newblock Two-sided {Bogoliubov} inequality to estimate finite size effects in quantum molecular simulations.
\newblock {\em Lett. Math. Phys.}, 112:97, 2022.

\bibitem{Reible2025apx}
B.~M. {Reible}, C.~Hartmann, and L.~{Delle Site}.
\newblock Finite-size effects in molecular simulations: a physico-mathematical view.
\newblock {\em Adv. Phys. X}, 10(1):2495151, 2025.

\bibitem{Reible2023}
B.~M. Reible, J.~F. Hille, C.~Hartmann, and L.~{Delle Site}.
\newblock Finite-size effects and thermodynamic accuracy in many-particle systems.
\newblock {\em Phys. Rev. Res.}, 5(2):023156, 2023.

\bibitem{RivasHuelga2012}
A.~Rivas and S.~F. Huelga.
\newblock {\em Open Quantum Systems}.
\newblock Springer, Berlin, Heidelberg, 2012.

\bibitem{Rivas2010}
A.~Rivas, A.~D.~K. Plato, S.~F. Huelga, and M.~B. Plenio.
\newblock Markovian master equations: a critical study.
\newblock {\em New J. Phys.}, 12:113032, 2010.

\bibitem{Rotter2015}
I.~Rotter and J.~P. Bird.
\newblock A review of progress in the physics of open quantum systems: theory and experiment.
\newblock {\em Rep. Prog. Phys.}, 78(11):114001, 2015.

\bibitem{Ruelle1999}
D.~Ruelle.
\newblock {\em Statistical Mechanics: Rigorous Results}.
\newblock Imperial College Press and World Scientific, London, Singapore, 1999.

\bibitem{Schmuedgen2012}
K.~Schmüdgen.
\newblock {\em Unbounded Self-adjoint Operators on Hilbert Space}.
\newblock Graduate Texts in Mathematics 265. Springer, Dordrecht, 2012.

\bibitem{Schwabl2006}
F.~Schwabl.
\newblock {\em Statistical Mechanics}.
\newblock Springer, Berlin, Heidelberg, {S}econd edition, 2006.

\bibitem{Simon1971paper}
B.~Simon.
\newblock Hamiltonians defined as quadratic forms.
\newblock {\em Commun. Math. Phys.}, 21:192–210, 1971.

\bibitem{Simon1971book}
B.~Simon.
\newblock {\em Quantum Mechanics for Hamiltonians Defined as Quadratic Forms}.
\newblock Princeton University Press, Princeton, NJ, 1971.

\bibitem{Simon1983}
B.~Simon.
\newblock Some quantum operators with discrete spectrum but classically continuous spectrum.
\newblock {\em Ann. Phys.}, 146(1):209--220, 1983.

\bibitem{Strocchi2013}
F.~Strocchi.
\newblock {\em An Introduction to Non-Perturbative Foundations of Quantum Field Theory}.
\newblock International Series of Monographs on Physics 158. Oxford University Press, Oxford, UK, 2013.

\bibitem{Teschl2014}
G.~Teschl.
\newblock {\em Mathematical Methods in Quantum Mechanics}.
\newblock Graduate Studies in Mathematics 157. American Mathematical Society, Providence, RI, {S}econd edition, 2014.

\bibitem{Trudinger1997}
N.~S. Trudinger.
\newblock On new isoperimetric inequalities and symmetrization.
\newblock {\em J. reine angew. Math.}, 488:203--220, 1997.

\bibitem{Vacchini2024}
B.~Vacchini.
\newblock {\em Open Quantum Systems}.
\newblock Graduate Texts in Physics. Springer, Cham, 2024.

\bibitem{vanNeerven2022}
J.~van Neerven.
\newblock {\em Functional Analysis}.
\newblock Cambridge Studies in Advanced Mathematics 201. Cambridge University Press, Cambridge, UK, 2022.

\bibitem{Zagrebnov2001}
V.~A. Zagrebnov and J.-B. Bru.
\newblock The {B}ogoliubov model of weakly imperfect {B}ose gas.
\newblock {\em Phys. Rep.}, 350(5-6):291--434, 2001.

\end{thebibliography}

\end{document}